\documentclass[10pt,twocolumn,twoside,aps,prl,showpacs,superscriptaddress,a4paper]{revtex4-1}

\usepackage{graphicx}
\usepackage{amsmath,amssymb}
\usepackage{color}
\usepackage{bbm}
\usepackage{hyperref}
\usepackage{theorem}
\usepackage{latexsym}

\newtheorem{theorem}{Theorem}

\newtheorem{corollary}[theorem]{Corollary}

\newtheorem{definition}[theorem]{Definition}

\newtheorem{lemma}[theorem]{Lemma}

\newlength{\blank}
\newenvironment{proof}[1][{\hspace{-\blank}}]{{\noindent\textbf{Proof~{#1}.\ }}}{\hfill\qed}
\newenvironment{proofthm}[1]{{\noindent\textbf{Proof~{#1}\ }}}{\hfill\qed}

\newcommand{\ket}[1]{|#1\rangle}
\newcommand{\bra}[1]{\langle#1|}
\newcommand{\braket}[2]{\langle #1 | #2 \rangle}

\mathchardef\ordinarycolon\mathcode`\:
\mathcode`\:=\string"8000
\def\vcentcolon{\mathrel{\mathop\ordinarycolon}}
\begingroup \catcode`\:=\active
  \lowercase{\endgroup
  \let :\vcentcolon
  }

\def\bea{\begin{array}}
\def\eea{\end{array}}

\newcommand{\nc}{\newcommand}
\nc{\rnc}{\renewcommand}
\nc{\be}{\begin{equation}}
\nc{\ee}{\end{equation}}
\nc{\ben}{\begin{eqnarray}}
\nc{\een}{\end{eqnarray}}
\nc{\lbar}[1]{\overline{#1}}
\nc{\ketbra}[2]{|#1\rangle\!\langle#2|}
\nc{\proj}[1]{| #1\rangle\!\langle #1 |}
\nc{\avg}[1]{\langle#1\rangle}
\nc{\Rank}{\operatorname{Rank}}
\nc{\smfrac}[2]{\mbox{$\frac{#1}{#2}$}}
\nc{\tr}{\operatorname{Tr}}
\nc{\ox}{\otimes}
\nc{\dg}{\dagger}
\nc{\dn}{\downarrow}
\nc{\cA}{\mathcal{A}}
\nc{\cB}{\mathcal{B}}
\nc{\cC}{\mathcal{C}}
\nc{\cD}{\mathcal{D}}
\nc{\cE}{\mathcal{E}}
\nc{\cF}{\mathcal{F}}
\nc{\cG}{\mathcal{G}}
\nc{\cH}{\mathcal{H}}
\nc{\cI}{\mathcal{I}}
\nc{\cJ}{\mathcal{J}}
\nc{\cK}{\mathcal{K}}
\nc{\cL}{\mathcal{L}}
\nc{\cM}{\mathcal{M}}
\nc{\cN}{\mathcal{N}}
\nc{\cO}{\mathcal{O}}
\nc{\cP}{\mathcal{P}}
\nc{\cR}{\mathcal{R}}
\nc{\cS}{\mathcal{S}}
\nc{\cT}{\mathcal{T}}
\nc{\cX}{\mathcal{X}}
\nc{\cZ}{\mathcal{Z}}
\nc{\csupp}{{\operatorname{csupp}}}
\nc{\qsupp}{{\operatorname{qsupp}}}
\nc{\var}{\operatorname{var}}
\nc{\rar}{\rightarrow}
\nc{\lrar}{\longrightarrow}
\nc{\polylog}{\operatorname{polylog}}
\nc{\1}{{\openone}}
\nc{\id}{{\operatorname{id}}}

\nc{\RR}{{{\mathbb R}}}
\nc{\CC}{{{\mathbb C}}}
\nc{\FF}{{{\mathbb F}}}
\nc{\NN}{{{\mathbb N}}}
\nc{\ZZ}{{{\mathbb Z}}}
\nc{\PP}{{{\mathbb P}}}
\nc{\QQ}{{{\mathbb Q}}}
\nc{\UU}{{{\mathbb U}}}
\nc{\EE}{{{\mathbb E}}}

\nc{\qed}{{\hfill$\Box$}}

\def\>{\rangle}
\def\<{\langle}

\begin{document}

\title{Distributed private randomness distillation}

\author{Dong Yang}
\email{dyang@cjlu.edu.cn}
\affiliation{Laboratory for Quantum Information, China Jiliang University, 310018 Hangzhou, China}
\affiliation{Department of Informatics, University of Bergen, 5020 Bergen, Norway}

\author{Karol Horodecki}
\email{khorodec@inf.ug.edu.pl}
\affiliation{International Centre for Theory of Quantum Technologies, University of
Gda\'{n}sk, Wita Stwosza 63, 80-308 Gda\'{n}sk, Poland}
\affiliation{Institute of Informatics, Department of Physics, Mathematics and Informatics, National Quantum Information Centre, University of Gda\'{n}sk, 80-308 Gda\'{n}sk, Poland}

\author{Andreas Winter}
\email{andreas.winter@uab.cat}
\affiliation{ICREA \& 
F\'{\i}sica Te\`{o}rica: Informaci\'{o} i Fen\`{o}mens Qu\`{a}ntics, 
Departament de F\'{\i}sica, Universitat Aut\`{o}noma de Barcelona, 
08193 Bellaterra (Barcelona), Spain}

\date{30 August 2019}

\begin{abstract}
We develop the resource theory of private randomness extraction in 
the distributed and device-dependent scenario. 
We begin by introducing the notion of independent random bits, which 
are bipartite states containing ideal private randomness for each party, 
and motivate the natural set of free operations. As a conceptual tool, we introduce \emph{Virtual Quantum State Merging}, 
which is essentially the flip side of Quantum State Merging, without
communication. 
We focus on the bipartite case and find the rate regions achievable in 
different settings. Surprisingly, it turns out that local noise 
can boost randomness extraction. 
As a consequence of our analysis, we resolve a long-standing problem by giving 
an operational interpretation for the reverse coherent information (up to a constant 
term $\log d$) as the number of private random bits obtained by sending quantum 
states from one honest party (server) to another one (client) via the eavesdropped 
quantum channel. 
\end{abstract}

\maketitle

%%%%%%%%%%%%%%%%%%%%%%%
{\it Introduction.}--%
Randomness is an important notion, having various applications 
in science and technology. Usually only pseudo-randomness is produced 
in the classical world, for instance by certain complex algorithms in a 
computer, where the pseudo-random value is determined by a hidden variable 
so that it is already implicitly known beforehand. 
Conceptually, the most straightforward way to ensure
a uniformly random bit sequence is to generate it by measuring a quantum
state, e.g., the $\sigma_Z$ eigenstate $\ket{0}$ in the $\sigma_X$
eigenbasis $\ket{\pm}$. In this way the measurement outcome is completely 
unpredictable and, thus, private against any eavesdropper. Here the privacy comes from the fact that a pure state naturally excludes any correlation with other systems. 
The problem of randomness extraction from a general mixed state  
has been considered in \cite{BFW} by the decoupling approach \cite{merging,decoupling,decouple}, where Alice's system is correlated with the system of an eavesdropper Eve via a mixed state $\rho_{AE}$ and the goal of Alice is to generate randomness private against Eve. An implicit assumption in this setting is that Bob, who holds 
the purifying system of $\rho_{AE}$, is a trusted but otherwise completely 
passive party. The reason is that Alice herself cannot figure out the correlation with Eve without Bob's assistance. So in the spirit of being cautious in cryptography, we have to assume that Eve holds all the purifying system of $\rho_A$, i.e. $\rho_{AE}$ is pure, unless Alice knows that Bob holds part of it. 

%\medskip
We make Bob active, where Alice and Bob trust each other and collaborate to  extract independent randomness private against Eve. This is a novel scenario: distributed private randomness extraction, which is dual to the Slepian-Wolf problem of distributed data compression \cite{SW} in information theory. Surprisingly, a natural dual setting to Slepian-Wolf does not exist in the classical framework, yet quantumly it does.
In this Letter, we study the distributed and device-dependent
scenario for randomness extraction (for an alternative approach, the so-called device-independent
scenario, see \cite{Bera et al} and references therein). 
We begin with defining the notion of
independent random bits (\emph{ibits}), in a picture dual to the standard one, as 
bipartite states that contain ideal private randomness, and justifying 
the set of allowed operations which do not increase randomness. 
Then we introduce our conceptual tool, the 
Virtual Quantum State Merging (VQSM) protocol, to study two-sided and one-sided 
randomness extraction. VQSM originates from the 
Quantum State Merging (QSM) protocol \cite{merging} and represents the other face of 
QSM, less noticed in the literature. In the two-sided setting, we obtain the achievable 
rate regions in various scenarios, including either free or no communication
and either free or no local noise. It follows that there is no bound randomness.
Surprisingly, local noise, usually regarded as useless,
can extend the rate region. In the one-sided setting, we determine the optimal rates 
of randomness extraction in two extremal classes, pure entangled 
states and separable states, and provide a computable 
upper bound for general states.  
Finally, we resolve a long-standing problem by giving an operational 
interpretation for the reverse coherent information in terms of the number of ibits obtained by sending quantum 
state from one honest party to another one via an 
eavesdropped quantum channel. In the following, we state and discuss the 
results carefully, while all proofs are in the appendix \cite{SM}.

\medskip
{\it Ibits and CLODCC.}--% 
In the standard approach, a state having one ideal random bit with respect to Eve has the form 
$\rho_{K_AE}=\frac{1}{2}(\ketbra{0}{0}\!+\!\ketbra{1}{1})_{K_A}\otimes \rho_E$ where $\rho_E$ is an arbitrary state of Eve. 
In the dual picture, where the purifying system is included and Eve's system is excluded, 
an equivalent form is 
$\alpha_{K_AA'B'} = {1\over 2}\sum_{i,j=0}^{1} \ketbra{i}{j}_{K_A}\otimes U_i \sigma_{A'B'} U_j^{\dagger}$, 
where $K_A$ is the key part that generates randomness if a measurement were 
performed in the basis $\{\ket{0},\ket{1}\}_{K_A}$, $U_i$ are unitary operators on $A'B'$,
and $\sigma_{A'B'}$ is an arbitrary state. $A'B'$ is known as the \emph{shield system} 
possibly distributed over Alice's and Bob's spaces, 
protecting privacy against Eve \cite{pptkey}. Similarly, a state with two independent random bits at Alice's and Bob's 
side, respectively, called \emph{ibit}, has the following form:
{\lemma \label{ibit-form} 
A bipartite quantum state having two independent random bits private against Eve is of the form
\be\label{alpha-state}
\alpha_{K_AK_BA'B'} = \frac14 \sum_{i,j,k,\ell=0}^{1} 
                              \ketbra{i}{j}_{K_A}\otimes \ketbra{k}{\ell}_{K_B} 
                               \otimes U_{ik} \sigma_{A'B'} U_{j\ell}^{\dagger}.
\ee
}

It is important to note that getting (approximate) ibits is equivalent to obtaining the state in the exact (approximate) standard form:
\be
\rho_{K_AK_BE}=\frac{1}{4}(\ketbra{0}{0}\!+\!\ketbra{1}{1})_{K_A}\otimes (\ketbra{0}{0}\!+\!\ketbra{1}{1})_{K_B}\otimes\rho_E,
\ee
see the proof of Lemma \ref{ibit-form} in \cite{SM}. The security of randomness is measured by the trace distance \cite{Christandl-thesis} and, thus, is composable \cite{Ben,RK}. 

Now we ask what kind of operations are allowed for free in randomness extraction. Notice that a pure state $\ket{0}_{K_A}$ is a special form of ibit where $K_B$ and $A'B'$ are dimension one, so we cannot allow pure states for free. It is safe to assume a closed system paradigm, like the framework for distilling thermodynamical work represented by pure states \cite{OHHH2001}. Also it is natural to assume that free operations should allow for local unitary transformation and some form of communication. A good candidate is the set of operations formed by \emph{Closed Local Operations} (CLO) and \emph{Dephasing Channel Communication} (DCC) \cite{OHHH2001,huge-delta}, designed for quantifying the localizable purity in a quantum state, where DCC simulates classical communication. In our setting, partial trace is not needed since we can always put the partial-tracing subsystem into the shield without harming privacy. In total, the set of free operations consists of the following two and their compositions:
(i) local unitary transformations and
(ii) sending a subsystem through a dephasing channel, where the dephasing channel environment goes to Eve. W.l.o.g.~the dephasing basis is chosen to be the fixed computational one and these operations are named \emph{CLODCC}. By Lemma \ref{ibit-form}, the picture of distilling ibits without measurement is equivalent to that of getting randomness in the standard form after performing measurement on the key part. We can therefore exchange both pictures freely. 

It might seem that our framework for randomness reduces to that for 
purity since we have similar operations and purity can generate randomness. 
This would be true if we considered randomness extraction under global operations, 
but it is very different in the distributed setting on which we focus in this work. 
Also, it is possible to study randomness extraction by assuming other 
free operations, e.g., incoherent operations \cite{HZ}.  

Having defined the free operations, we now ask the key question: For a bipartite quantum state $\rho_{AB}$ whose purification is with Eve, how much private randomness can Alice and Bob obtain against Eve under CLODCC? We mainly consider the asymptotic i.i.d. setting which means we count the rates. It turns out that another seemingly useless resource--local noise, can act as a booster in the process. Here local noise means a maximally mixed state on Alice's or Bob's side, whose purification is under Eve's control, i.e. Alice or Bob share a maximally entangled state with Eve. It is clear that from local noise alone, Alice cannot produce randomness unknown to Eve by measuring her half, because the outcome is perfectly correlated with Eve. However local noise may help when combined with other states. An illuminating example is entanglement swapping: Alice shares one singlet with Bob and another one with Eve, i.e. the tripartite state is $\ket{\Phi}_{A_1B}\otimes\ket{\Phi}_{A_2E}$. Here the state $\frac{1}{2}\1_{A_2}$ is understood as the local noise on Alice's side. Observe that in entanglement swapping, the outcome of the Bell measurement by Alice is completely random against each of Bob and Eve separately (no communication between them). Thus we get that in the case of no communication between Alice and Bob, Alice can obtain two random bits unknown to Eve. However without the local noise, 
Alice can get only one random bit. We therefore have several different settings 
depending on whether randomness is to be distilled on two sides or one side, 
whether local noise is available or not, and whether communication is allowed or not. 
Before stating our findings on the rate regions,
we introduce the conceptual tool.

\medskip
%%%%%%%%%%%%%%%%%%%%%%%%
{\it VQSM.}--Entanglement swapping shows that local noise can play an important role 
in the distributed scenario. If local noise is not freely available, Alice and Bob have the option of creating it from the resource. 
In the case of no communication, Alice can put one copy of her systems $A$ into 
the shield $A'$, meaning that it will not be touched by her any more. This does not change Bob's state. We may pretend that $A'$ is with Eve, making things only worse for Bob, but now it is a pure state with Eve, the definition of local noise. In the case of communication, Alice can send one of her systems to Bob through the dephasing channel. Since randomness in the copy is not extracted, these options may reduce the overall rate of randomness distillation if the proportion of the wasted copies is not negligible. 
However, there is a more efficient way, which produces some private
randomness for Alice and simultaneously gives Bob local noise. To gain 
intuition, let us look at the QSM protocol \cite{merging}.
For our goal, B is the reference system and QSM is performed from A to E, 
i.e., the task is to transform $n$ copies of a tripartite state $\ket{\psi}_{ABE}$ into 
$\ket{\phi}_{A_1A_2B^nE'E_1E_2}$ by LOCC such that 
$\phi_{E_1B^nE'}\approx(\ketbra{\psi}{\psi}_{ABE})^{\otimes n}$, 
where $\phi_{E_1B^nE'}=\tr_{A_1A_2E_2}\phi_{A_1A_2B^nE'E_1E_2}$. 
In the asymptotic i.i.d. case, if 
$S(A|E)_{\psi} = S(\psi_{AE})-S(\psi_E)>0$, the protocol requires 
an additional rate of $S(A|E)_{\psi}$ ebits (the unit entanglement in a two-qubit maximally entangled state) shared between systems A and E. 
However, when $S(A|E)_{\psi}<0$, not only does it not need entanglement, 
but also creates a rate of $-S(A|E)_{\psi}$ ebits on $A_2E_2$.
In the protocol, Alice transforms $A^n$ into $A_1A_2$ by a unitary, 
then measures subsystem $A_1$ in the computational basis and announces 
the outcome to Eve, who further, according to the outcomes, transforms $E^n$ into $E'E_1E_2$ 
by proper unitaries. The amount of classical communication
is the size of $A_1$, its rate is the mutual information between Alice and Bob,
$I(A:B)_\psi = S(\psi_A)+S(\psi_B)-S(\psi_{AB})$.

Roughly speaking, the measurement outcome on $A_1$ in QSM is independent of Bob and private against Eve. The intuition is the following: Bob's state remains invariant, so the measurement outcome is independent of Bob; the necessity to send the measurement outcome from Alice to Eve implies that Eve cannot predict it by herself and, thus, it is private against Eve. Technically, the measurement outcome is almost decoupled from Eve. This weak correlation can be deleted by the technique of 
privacy amplification (PA) \cite{RK}; 
it amounts to tracing out a little bit more from $A_1$, not affecting the rate of randomness. 
A further observation is that Bob's system is purely entangled with 
Eve's, conditional on the measurement outcomes, and in our setting, Bob need not care about 
whether Eve performs the rotations or not. In this way, Alice can extract randomness independent of Bob and 
private against Eve at the rate $I(A:B)_{\psi}$, while at the same time 
Bob's system is virtually entangled with Eve, acting as
noise that can help Bob extract randomness later. The composition of QSM and PA we call \emph{Virtual Quantum State Merging}. 
The formal claim that the rate of randomness extractable by Alice in the QSM setting equals to $I(A:B)_{\psi}$ is encapsulated in a double-decoupling theorem, \cite[Theorem \ref{two-decouple}]{SM}. Its composability, i.e. Alice's and Bob's randomness is independent when Bob extracts randomness later assisted by local noise, is proved in \cite[Theorem \ref{security}]{SM}. Theorem \ref{two-decouple}, which may be of independent interest itself, is our main tool to derive randomness distillation rates.

%%%%%%%%%%%%%%%%%%%%%%
\medskip
{\it Two-sided randomness distillation.}--% 
Having developed the tool, we are ready to state the main results.  
It is clear that given a state $\rho_{AB}$, we can 
obtain the randomness at the rate $R_G(\rho_{AB}):=\log|AB|-S(\rho_{AB})$ if global 
unitary operations on $AB$ are allowed \cite[Lemma \ref{global-purity}]{SM}. 
We now ask the same question when the parties are distributed so only local 
unitary operations are allowed: What is the rate region of achievable pairs
$(R_A, R_B)$, representing that 
Alice produces randomness at rate $R_A$ and Bob at rate $R_B$, and their 
randomness is independent and secret against Eve, who has the purifying system?
We have four different settings in which we allow free or no local noise,
and free or no communication in form of the dephasing channel. 
The answer is our first main result, Theorem \ref{two-side}.

{\theorem \label{two-side} 
For a given state $\rho_{AB}$, the following rate regions are achievable
(and tight in settings 1, 2 and 3):

1) For no communication and no noise, 
$  R_A       \leq \log|A| - S(A|B)_+ $, 
$  R_B       \leq \log|B| - S(B|A)_+ $, and
$ R_A + R_B \leq R_G$,
where $[t]_+=\max\{0,t\}$;

2) for free noise but no communication, 
$R_A       \leq \log|A| - S(A|B)$, 
$R_B       \leq \log|B| - S(B|A)$, and 
$R_A + R_B \leq R_G$;

3) for free noise and free communication, 
$R_A\leq R_G$, $R_B\leq R_G$, and $R_A+R_B\leq R_G$;

4) for free communication but no noise, 
$R_A       \leq \log|AB| - \max\{S(B),S(AB)\}$, 
$R_B       \leq \log|AB| - \max\{S(A),S(AB)\}$, and
$R_A + R_B \leq R_G$.}

Note how the solutions to 1) and 2) appear to be dual to the Slepian-Wolf theorem on distributed data compression \cite{SW}. In 3), the rate $R_G$ can be realised on either side as randomness (but not necessarily as purity). We prove only achievability in 4) and leave its tightness open. 

%%%%%%%%%%%%%%%%%%%%
\begin{figure}[ht]
\includegraphics[width=0.36\textwidth]{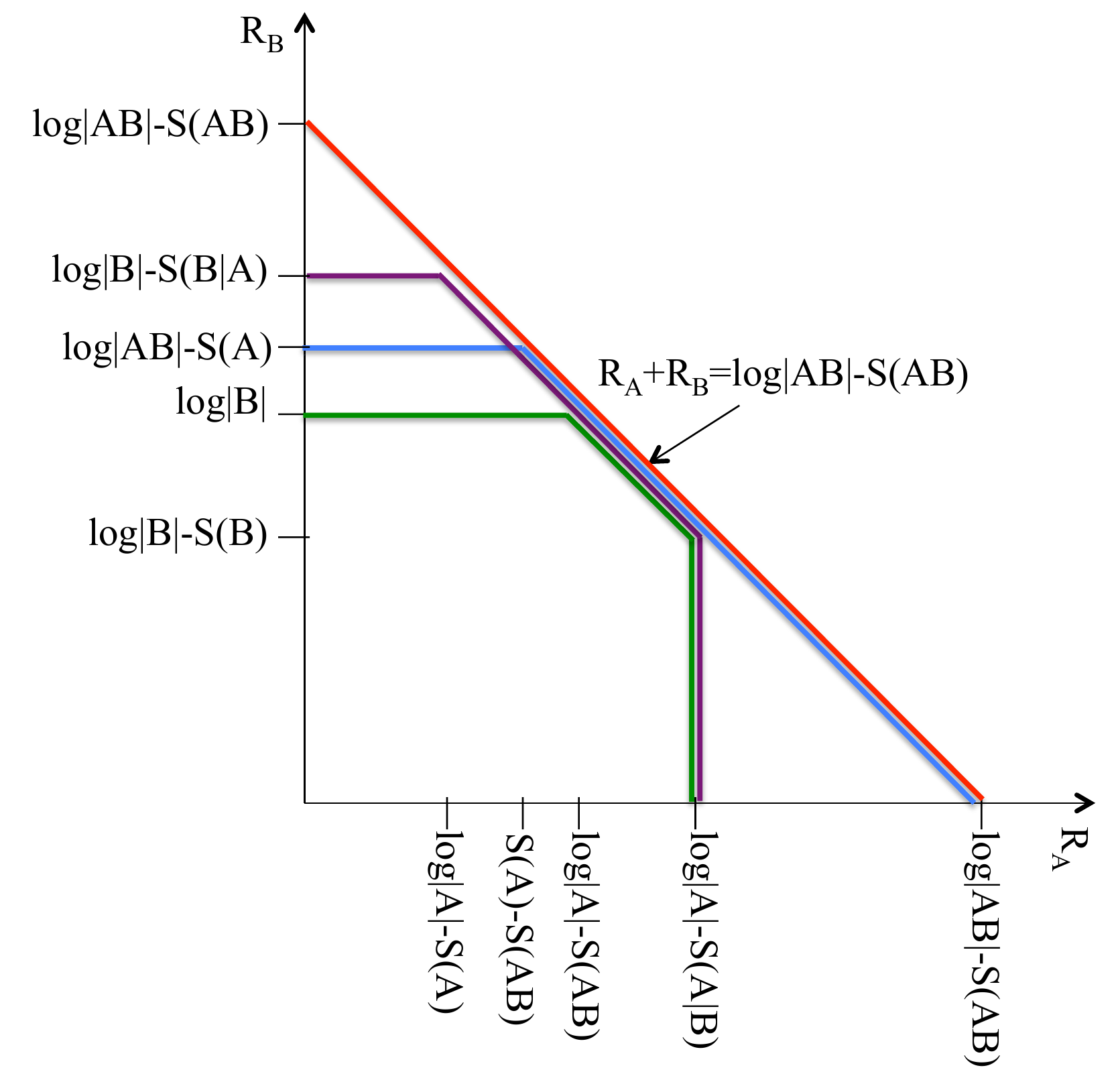}
\caption{The rate regions of $(R_A, R_B)$ when $S(A|B)>0>S(A)-\log|A|>S(B|A)$. The green lines show the rate region of setting 1), the purple lines of setting 2), the red lines of setting 3), and the blue lines of setting 4).}\label{fig1}
\end{figure}

It is easy to see that, if non-zero randomness can be extracted by global operations, which 
means $R_G>0$, then from Theorem \ref{two-side} there exists a pair $(R_A,R_B)$ satisfying $R_A+R_B=R_G>0$ achievable even in the most restrictive setting 1). Therefore, an immediate conclusion from Theorem \ref{two-side} is the non-existence of bound randomness states.

{\corollary \label{no-bound} Given a state $\rho_{AB}$, if randomness can be extracted by 
global operations on $AB$, then non-zero randomness can be extracted by CLODCC operations}.

\medskip
%%%%%%%%%%%%%%%%%%%%%
{\it One-sided randomness distillation.}--% 
We will focus on a pair rate $(R_A,0)$ in setting 4) to understand it a bit better. The task is that Bob helps Alice to extract randomness against 
Eve as much as possible. Free communication is allowed but local noise is not. Other 
settings are understood already in Theorem \ref{two-side} from the tightness of the rate regions, which imply that the extremal points are optimal. First, we derive a formula for the optimal rate, unfortunately involving regularization.  

{\theorem \label{randomness-one}
The randomness that can be extracted from $\rho_{AB}$ on Alice's side is
$
R_A(\rho)=\log|AB|-\inf\frac1n\max\bigl\{ S(E'^{(n)}), S(B'^{(n)}) \bigr\},
$
where the infimum is taken over all $n$, and $E'^{(n)}$ and $B'^{(n)}$ are the output
systems under CLODCC acting on $\rho^{\otimes n}_{AB}$.}

%\medskip
At first sight, this theorem is not very useful as it involves complex CLODCC processing 
and suffers from the notorious regularization problem. 
But the further observation that 
$S(E'^{(n)})\ge nS(E)$, because Eve's entropy is non-decreasing in every use of the 
dephasing channel, gives the exact optimal rate for a large class of states. Namely,
for $\rho_{AB}$ satisfying $S(B)\le S(E)$, $R_A(\rho)=\log|AB|-S(AB)=R_G$, 
so all the randomness can be localized on Alice's side. A large class of bipartite states satisfy this property, e.g. the 
positive-partial-transpose states, which include separable states. 
Note that possibly two-way communication is needed to localize the global purity on Alice's 
side, but only one-way communication is needed to localize randomness. 
An example for this is the cq-state 
$\rho_{AB}=\sum_i p_i\ketbra{i}{i}_A\otimes \rho^i_B$, when $\rho^i_B$ do not commute.  

Whether the regularization is needed or not is an open problem. We give a partial answer 
to this in that $R_A$ is at least not strongly additive, i.e. 
$R_A(\rho\otimes \sigma)\neq R_A(\rho)+R_{A}(\sigma)$, by presenting an 
example with the activation effect: $R_A(\Phi_{AB}\otimes \sigma_{AB})=3$ while 
$R_A(\Phi_{AB})=\frac32$ and $R_A(\sigma_{AB})=1$, where $\Phi_{AB}$ is a singlet 
and $\sigma_{AB}=\frac12(\ketbra{00}{00}\!+\!\ketbra{11}{11})$. The nontrivial
$R_A(\Phi_{AB})=\frac32$ comes from Theorem \ref{uppbound-one}, and the other two are
straightforward. 

{\theorem \label{uppbound-one}
The optimal randomness that can be extracted from $\rho_{AB}$ on Alice's 
side is upper bounded by $R_A(\rho)\leq \log|AB|-\frac12\max\{S(A),S(B)\}$. The upper bound is tight for pure states.
}

%\medskip

Hence, for $\ket{\phi}_{AB}$, $R_A=\log|AB|-\frac12S(A)_{\phi}$. Furthermore, the rate of the secret key that can be distilled from the pure state is $K_D=S(A)_{\phi}$ \cite{pptkey}. When $|A|=|B|=d$, we have an appealing formula exhibiting 
the exact balance between localisable and shareable privacy
$R_A+\frac12K_D=2\log d$.
For mixed states $K_D(\rho_{AB})\le S(\rho_A)$, 
we get $R_A+\frac12K_D\le 2\log d$,
which can be treated as complementarity between private
randomness and key in analogy to complementarity between purity and 
entanglement \cite{JO}. 
It captures the fact that $R_A$ attains $2\log d$ on a pure product state 
while $K_D$ attains the maximal value on a maximally entangled state.

\medskip
%%%%%%%%%%%%%%%%%%%%
{\it Private randomness capacity.}--% 
It is well known that channel capacity formulas usually involve regularization 
\cite{Smith}, i.e., optimisation over growing numbers of channel uses due to 
non-additivity of the relevant quantities 
\cite{Hastings,DiVincenzo-Qnonadditive,Li-Winter,Smith-privacy}, making a 
head-on numerical approach impossible. However, this does not prevent the 
regularized non-additive quantities from having physical interpretations: the optimal rate to transmit information faithfully through a channel. Ironically, there are additive quantities that however
lack an interpretation \cite{uni-additivity}, one of these being the 
\emph{reverse coherent information}, which is additive 
\cite{DJKR}, but whose interpretation was missing for a long time. The quantity was defined in \cite{Hayashi} and introduced in \cite{DJKR} as the ``negative cb-entropy'' of a channel.
In \cite{reverse-coh}, it was rediscovered independently as {\it reverse coherent information} (RCI) and 
shown to be a lower bound for the entanglement 
distribution capacity of a channel assisted by classical feedback communication by using the hashing inequality \cite{DevetakWinter-hash} (see \cite{MPRHorodecki} in this context). See  \cite{continuous-1,continuous-2} for the extension to the continuous-variable case and the very recent papers on RCI \cite{SKG,GW,KW,Junge1,Junge2,GES,VA,Pirandola3,Pirandola5}. Here we provide its exact operational interpretation.  

Consider the task of generating private randomness by communicating 
through a quantum channel. A client Bob wanting to produce private randomness, has the measurement device but cannot prepare a quantum state himself. 
He has access to but does not trust quantum systems being at his disposal, which are potentially entangled with Eve's system. However, there is a quantum channel to Bob from a trusted server Alice who can prepare any quantum state, although the channel itself is eavesdropped by Eve. 
First, we argue that the model is well motivated and that it captures the server-client structure possibly realised in future 
quantum networks: the server having huge devices and being able to 
prepare and manipulate quantum states and a client being able to only perform 
limited operations, e.g., unitaries and measurement. 
The server provides a service to the client via a quantum channel. 
Second, this new cryptographic model is completely in line with the standard model of 
transmitting information and, thus, can be regarded 
as a new character of a quantum channel. Now the natural capacity question is: 
``What is the maximal rate of private randomness that can be extracted at Bob's side when 
Alice sends quantum states through the channel?'' Amazingly, we can answer
it completely.

{\theorem \label{randomness-capacity}
The private randomness capacity of a channel 
$\mathcal{N}:A'\longrightarrow B$ is given by\\
$
  R({\cal N}) = \log|B| 
               + \max_{|\phi\>_{AA'}} \bigl\{ S(\phi_A)-S(\id_A\otimes{\cal N}(\phi_{AA'})) \bigr\}.
$}

%\medskip
The second term is just the reverse coherent information of the channel. So, for the first time, we provide an operational interpretation for it in a
natural Shannon-theoretic model. [After our work, Mark Wilde pointed out that the same quantity $R({\cal N})$ can be derived from \cite[Theorem 15]{hypothesis} and, thus, can be explained in quantum hypothesis testing for channels.] As a matter of fact, the randomness generated at the client is private not only to the eavesdropper but also to the server. 
The single-letter formula is concave w.r.t. the input state \cite{Holevo} and, thus, efficiently computable. 
We can draw an interesting comparison between private randomness capacity and 
purity capacity: 
$\text{Pure}({\cal N})=\log|B|-\inf_n \frac1n S_{\min}({\cal N}^{\otimes n})$, where 
$S_{\min}({\cal N})=\min_{\phi}S({\cal N}(\phi))$ is the minimum output entropy. 
From \cite{Hastings}, we know that $S_{\min}({\cal N})$ is not additive; thus regularisation is required.

\medskip
%%%%%%%%%%%%%%%%%%%
{\it Summary and outlook.}--% 
We have initiated the study of randomness extraction in the distributed scenario and provided the tool to tackle this situation. We found the exact achievable rate regions for various settings of distillation protocols and gave the long-sought operational interpretation
of the reverse coherent information of a quantum channel.

This work opens up a new area, suggesting a wide range of generalisations 
and extension of the considered scenarios, including the multipartite case and 
the one-shot scenario based on one-shot QSM \cite{Berta-ms},
as well as other cryptographic variants such as the honest-but-curious scenario. The structural analogy between private states and independent states also needs further exploration. These directions will be developed elsewhere \cite{elsewhere}. Two notable topics are under study. One is distributed randomness extraction under the framework of coherence theory \cite{Aberg,BCP,WY}, where the allowed operations are incoherent, and the other is the strong converse problem for the private randomness capacity. Indeed, there are rare cases \cite{ent-assist,EBC} where we have a single-letter formula on capacity, and the strong converse theorem is proved only for such cases \cite{reverse-Shannon,WWY}.

\medskip
%%%%%%%%%%%%%%%%%%%
{\it Acknowledgements.}--%
The authors thank Omer Sakarya for confirming the direct relation between secret sharing schemes and the structure of independent states, and Mark Wilde and Stefano Pirandola for comments on an earlier version.
DY was supported by the NSFC (grant nos. 11375165, 11875244), and by the NFR project ES564777.
KH was supported by the grant Sonata Bis 5 (grant no. 2015/18/E/ST2/00327) from the National Science Center, and by the Foundation for Polish Science (IRAP project, ICTQT, contract no. 2018/MAB/5, co-financed by EU via Smart Growth Operational Programme).
AW was supported by the European Research Council (Advanced Grant IRQUAT), the European Commission (STREP project RAQUEL), the Spanish MINECO, projects FIS2013-40627-P and FIS2016-80681-P, as well as by the Generalitat de Catalunya, CIRIT projects 2014-SGR-966 and 2017-SGR-1127.

%%%%%%%%%%%%%%%%%%%

\clearpage

\appendix

\section*{SUPPLEMENTAL MATERIAL}

\section{A. Miscellaneous facts and lemmas}
In this appendix, we collect some standard facts about various functionals 
we use and prove some lemmas that we need in the proofs of main results.

Recall that the fidelity between two mixed states is defined as
\[
  F(\rho,\sigma): = \left\|\sqrt{\rho}\sqrt{\sigma}\right\|_1 
                  = \tr\sqrt{\sqrt{\rho}\sigma\sqrt{\rho}},
\]
and the trace distance is $\frac12\|\rho-\sigma\|_1$.
We use the notation $X\stackrel{\epsilon}{\approx}Y$ to mean $\|X-Y\|_1\le\epsilon$.

\begin{lemma}[Uhlmann \cite{Uhlmann}]
\label{Uhlmann}
The fidelity is alternatively characterized by the relation
\[
  F(\rho_A,\sigma_A) = \max_{U_B} \,\bigl|\bra{\phi}{\1_A\otimes U_B}\ket{\psi}\bigr|,
\]
where $\ket{\phi}_{AB}$ and $\ket{\psi}_{AB}$ are purifications of $\rho_A$ and 
$\sigma_A$, respectively, and $U_B$ ranges over unitaries.
\end{lemma}

\begin{lemma}[Fuchs and van de Graaf \cite{FG}] 
\label{FG}
For any two states $\rho$ and $\sigma$, 
fidelity and trace distance are related by
\[
  1-F(\rho,\sigma) \leq \frac12 \|\rho-\sigma\|_1 \leq \sqrt{1-F(\rho,\sigma)^2}.
\]
\end{lemma}

%%%%%%%%%%%%%%%%%%%%%%%
\begin{lemma}[{Winter}~\cite{gentle}]
\label{gentle-lemma}
For a subnormalized state $\rho$, i.e. $\rho\ge 0$ and $\tr\rho\le 1$, and operator $0\le X\le \1$, if $\tr\rho X\ge 1-\lambda$, then
\[
\|\rho-\sqrt{X}\rho\sqrt{X}\|_1\le\sqrt{8\lambda} .
\]
\end{lemma}

%%%%%%%%%%%%%%%%%%%%%%%
\begin{lemma}[{Alicki-Fannes}~\cite{continuity}]
\label{continuity-cond-entropy}
For two states $\rho_{XY}$ and $\sigma_{XY}$ on ${\cal H}_X\otimes {\cal H}_Y$, if $\|\rho_{XY}-\sigma_{XY}\|_1\le\lambda\le1$, then 
\[
|S(X|Y)_{\rho}-S(X|Y)_{\sigma}|\le 4\lambda\log|X|+2h(\lambda),
\]
where $h(\lambda):=-\lambda\log\lambda-(1-\lambda)\log(1-\lambda)$.
\end{lemma}

%%%%%%%%%%%%%%%%%%%%%%%
\begin{lemma}[Dupuis \emph{et al.} \cite{decouple}] 
  \label{Dupuis}
  Given a state $\rho_{AE}$, let $\epsilon>0$ and $\cT:{A\to B}$ be a 
  CPTP map with Choi-Jamio\l{}kowski representation 
  $\tau_{RB}=(\id_R\otimes\cT)\Phi_{RA}$. Then,
  \begin{equation}\begin{split}
    \int_{U_A} &\| \cT(U_A\rho_{AE}U_A^{\dagger}) - \tau_{B}\otimes\rho_E\|_1dU \\
               &\phantom{====}
                \le 2^{-\frac{1}{2}[H^{\epsilon}_{\min}(A|E)_{\rho}+H^{\epsilon}_{\min}(R|B)_{\tau}]}
                      +12\epsilon,
  \end{split}\end{equation}
  where the integral is over the Haar measure on the unitary group $\mathcal{U}(A)$.
\end{lemma}

%%%%%%%%%%%%%%%%%%%%%%%

\begin{lemma}[{Horodecki \emph{et al.}~\cite{huge-delta}}]
\label{key-ineq} 
The function defined by $g(\rho_{AB}):=S(\rho_{AB})+E_r(\rho_{AB})$ is non-decreasing under CLODCC.
\qed
\end{lemma} 

Here,
\[
  E_r(\rho_{AB}):=\min_{\sigma_{AB}\in \text{SEP}}S(\rho_{AB}\|\sigma_{AB})
\]
is the relative entropy of entanglement \cite{Er}, where the minimum is taken over 
the separable states (SEP), and $D(X\|Y)=\tr X(\log X-\log Y)$ is the relative entropy.

\medskip
%%%%%%%%%%%%%%%%%%%%%%%

%%%%%%%%%%%%%%%%%%%%%%%
%In the proof of Lemma \ref{preprocess}, we need the following slight modification of \cite[Lemma 1]{local-purity} and the gentle operator lemma 
%\cite{gentle}.
%\begin{lemma}[Cf. {Devetak}~\cite{local-purity}]
%Let $\Pi$ be a projector with $\tr\Pi=d_1$ and $\rho$ a state, both defined on a $d_1d_2$-dimensional Hilbert space $H_X$. If $\tr\rho\Pi\ge1-
%\epsilon$, then there exists a unitary $U: H_X\to H_Y\otimes H_Z$, with $dimH_Y = d_1$ and $dimH_Z = d_2$, such that
%\[
%\|U\rho U^{\dagger}-\frac{1}{\tr\Pi\rho}(\Pi\rho\Pi)_Y\otimes \ketbra{0}{0}_Z\|\le 2\epsilon.
%\]
%\end{lemma}
%\begin{proof}
%First notice that \cite[Lemma 1]{local-purity} is still true if the projector $\Pi$ does not commute with the state $\rho$. Then we use the triangle %inequality for the trace norm to arrive at the conclusion. 
%\end{proof}
%%%%%%%%%%%%%%%%%%%%%%%
\medskip
In the proof of Lemma \ref{preprocess}, we need the concept of typical operator and its properties. For a mixed state $\rho$, write it in its eigenbasis, $\rho=\sum_i\lambda_i\ketbra{i}{i}$. Consider $n$ copies of the state,
\[
\rho^{\otimes n}=\sum_{i^n}\lambda_{i^n}\ketbra{i^n}{i^n},
\]
where $i^n=i_1i_2\cdots i_n$. For $\delta>0$, the typical projector is defined as
\[
\Pi_{\delta}^{n}=\sum_{i^n\in{\cal T}_{\delta}^{n}}\ketbra{i^n}{i^n},
\]
where ${\cal T}_{\delta}^{n}:=\{i^n:|-\frac{1}{n}\log\lambda_{i^n}-S(\rho)|\le\delta\}$. The typical subspace is the supporting space of the the projector. From \cite{Schumacher,Winter-thesis}, we have
\begin{align}
\label{typ}
\tr\rho^{\otimes n}\Pi_{\delta}^{n}\ge1-\epsilon,\\
\label{typical-dimension}
2^{n[S(\rho)+\delta]}\ge\tr\Pi_{\delta}^{n}\ge(1-\epsilon)2^{n[S(\rho)-\delta]},
\end{align}
with $\epsilon:=e^{-c\delta^2n}$ and $c$ a constant.
 
For $n$ copies of a pure tripartite state $\ket{\psi}_{ABE}$, let ${\tilde{A}}$ and ${\tilde{B}}$ be the typical subspaces of $A^n$ and $B^n$ respectively, and $\Pi_{\tilde{A}}$ and $\Pi_{\tilde{B}}$ the typical projectors onto these typical subspaces, that is,
\begin{align*}
\tr\Pi_{\tilde{A}}\psi_A^{\otimes n}\ge1-\epsilon,\\
\tr\Pi_{\tilde{B}}\psi_B^{\otimes n}\ge1-\epsilon.
\end{align*}
Denote 
\begin{align*}
\ket{\Omega}_{\tilde{A}\tilde{B}E^n}&=\Pi_{\tilde{A}}\otimes\Pi_{\tilde{B}}\ket{\psi}_{ABE}^{\otimes n},\\
\ket{\Psi}_{\tilde{A}\tilde{B}E^n}&=\frac{1}{\braket{\Omega}{\Omega}}\ket{\Omega}_{\tilde{A}\tilde{B}E^n},
\end{align*}
and write $\Pi_{\tilde{A}}$ and $\Pi_{\tilde{B}}$ in their eigenbasis respectively
\begin{align*}
\Pi_{\tilde{A}}&=\sum_i \ketbra{i}{i}_{\tilde{A}},\\
\Pi_{\tilde{B}}&=\sum_j \ketbra{j}{j}_{\tilde{B}},
\end{align*}
where $\ket{i}_{\tilde{A}}\in {\cal H}_{A^n}$ and $\ket{j}_{\tilde{B}}\in {\cal H}_{B^n}$. By typicality in Eq. (\ref{typical-dimension}), we have
\begin{align*}
\frac{\log|\tilde{A}|}{n}&\approx {S(\psi_A)},\\
\frac{\log|\tilde{B}|}{n}&\approx {S(\psi_B)}.
\end{align*}
We now decompose ${\cal H}_{A^n}$ and ${\cal H}_{B^n}$ as
\begin{align*}
{\cal H}_{A^n}&={\cal H}_{A_I}\otimes {\cal H}_{A_P},\\
{\cal H}_{B^n}&={\cal H}_{B_I}\otimes {\cal H}_{B_P},
\end{align*}
by unitaries $U:A^n\to A_IA_P$ and $V:B^n\to B_IB_P$ satisfying
\begin{align}
\label{U}
U\ket{i}_{\tilde{A}}&=\ket{i}_{A_I}\otimes\ket{0}_{A_P},\\
\label{V}
V\ket{i}_{\tilde{B}}&=\ket{i}_{B_I}\otimes\ket{0}_{B_P}.
\end{align}
Notice that $|\tilde{A}|=|A_I|$ and $|\tilde{B}|=|B_I|$, so we get
\begin{align*}
U\otimes V\ket{\Psi}_{\tilde{A}\tilde{B}E^n}=\ket{\Psi}_{A_IB_IE^n}\otimes \ket{0}_{A_P}\otimes\ket{0}_{B_P}.
\end{align*}
We call $A_I$, $B_I$ the information parts, and $A_P$, $B_P$ the purity parts.
%%%%%%%%%%%%%%%%%%%%%%%
\begin{lemma}
\label{preprocess}
Given $\delta >0$ and $\epsilon=e^{-c\delta^2n}$ with a constant $c$, for $n$
copies of a pure tripartite state $\ket{\psi}_{ABE}$ where $n$ is large, and unitary operators $U$ and $V$ in Eqs. (\ref{U}) (\ref{V}), 
\begin{align}
\label{inf-pure}
&\left\|U\otimes V(\ketbra{\psi}{\psi})^{\otimes n} U^{\dagger}\otimes V^{\dagger}-\right.\nonumber\\
&~\left.\ketbra{\Psi}{\Psi}_{A_IB_IE^n}\otimes \ketbra{0}{0}_{A_P}\otimes\ketbra{0}{0}_{B_P}\right\|_1\le 2\epsilon+4\sqrt{\epsilon},
\end{align}
with $\frac{\log|A_I|}{n}=\frac{\log|\tilde{A}|}{n}\approx {S(\psi_A)}$, $\frac{\log|B_I|}{n}=\frac{\log|\tilde{B}|}{n}\approx {S(\psi_B)}$, and $\frac{\log|A_P|}{n}\approx \log |A|-S(\psi_A)$, $\frac{\log|B_P|}{n}\approx \log |B|-S(\psi_B)$, and all the entropy relations in $\ket{\Psi}_{A_IB_IE^n}$ is almost the same as those in $\ket{\psi}_{ABE}^{\otimes n}$.
\end{lemma}

\begin{proof}
By the operator inequality 
\[
\1_{A^n}\otimes\1_{B^n}-\Pi_{\tilde{A}}\otimes \Pi_{\tilde{B}}\le \1_{A^n}\otimes(\1_{B^n}-\Pi_{\tilde{B}})+(\1_{A^n}-\Pi_{\tilde{A}})\otimes\1_{B^n},
\]
we get
\[
\braket{\Omega}{\Omega}_{\tilde{A}\tilde{B}E^n}\ge 1-2\epsilon.
\]
By Lemma \ref{gentle-lemma}, we have
\[
\|(\ketbra{\psi}{\psi})^{\otimes n}-\ketbra{\Omega}{\Omega}\|_1\le4\sqrt{\epsilon},
\]
from which, 
\[
\|(\ketbra{\psi}{\psi}_{ABE})^{\otimes n}-(\ketbra{\Psi}{\Psi})_{\tilde{A}\tilde{B}E^n}\|_1\le 2\epsilon+4\sqrt{\epsilon}=:\epsilon'
\]
that implies Ineq. (\ref{inf-pure}) under unitary operators $U$ and $V$. 

Notice that $U$ and $V$ are local unitaries, thus the entropies and the conditional entropies in ${\psi}_{ABE}^{\otimes n}$ remain invariant. By the non-increasing property of trace norm for Ineq. (\ref{inf-pure}) and Lemma \ref{continuity-cond-entropy} (including entropies), we get 
\begin{align*}
|S(A|B)_{\psi^{\otimes n}}-S(A_I|B_I)_{\Psi}|\le 4n\epsilon'\log|A|+2h(\epsilon'),\\
|S(A)_{\psi^{\otimes n}}-S(A_I)_{\Psi}|\le 4n\epsilon'\log|A|+2h(\epsilon'),
\end{align*} 
and similar other relations. Notice that $\epsilon'$ is exponentially small in $n$, so we conclude that $\ket{\Psi}_{A_IB_IE^n}$ encodes almost all information in $\ket{\psi}_{ABE}^{\otimes n}$.

\end{proof}

%%%%%%%%%%%%%%%%%
\medskip
{\lemma \label{global-purity} For a state $\rho_{AB}$, if global operations are allowed, 
the randomness extraction rate is $R_G(\rho_{AB})=I_G(\rho_{AB})=\log|AB|-S(\rho_{AB})$.}

\medskip
\begin{proof}
It is clear that a pure state $\ket{0}$ is the simplest ibit. So given a state $\rho_{AB}$, 
if global operations are allowed, then we can obtain purity at the rate 
$I_G=\log|AB|-S(\rho_{AB})$, by data compression and algorithmic cooling \cite{huge-delta}.
From the purity the same amount of randomness can be obtained which implies $R_G\ge I_G$ 

On the other side, any randomness extraction ends up in a state close to $\alpha$ state of Eq. (\ref{alpha-state}), 
on which we can undo the twisting $U=\sum_{i,k}\ketbra{ik}{ik}_{K_AK_B}\otimes U_{ik}$ to get a pure state on the key part which means 
$I_G\ge R_G$. So we conclude $I_G=R_G$.
\end{proof}

%%%%%%%%%%%%%%%%%%%%%%%
\begin{lemma}
\label{markov}
If states in an ensemble $\{p_i,\rho_i\}$ are close to a fixed state $\sigma$ on average, then most of the ensemble states are close to the fixed state. To be precise, if 
\be
  \sum_i p_i\|\rho_i-\sigma\|_1\le\epsilon,
\ee
then
\be
  \sum_{i: \|\rho_i-\sigma\|_1\le\sqrt{\epsilon}} p_i\ge1-\sqrt{\epsilon}.
\ee
\end{lemma}
\begin{proof}
This is an instance of Markov's inequality. It is obvious that 
\be
  \sum_{i: \|\rho_i-\sigma\|_1\ge\sqrt{\epsilon}} \sqrt{\epsilon}p_i
    \le\sum_ {i: \|\rho_i-\sigma\|_1\ge\sqrt{\epsilon}} p_i\|\rho_i-\sigma\|_1 \le \epsilon,
\ee
thus
\be
  \sum_{\{i: \|\rho_i-\sigma\|_1\ge\sqrt{\epsilon}\}} p_i\le \sqrt{\epsilon},
\ee
concluding the proof.
\end{proof}

\medskip
%%%%%%%%%%%%%%%%%%%%%%%

\section{B. Proofs}
In this appendix, we provide the detailed proofs of the results claimed in the main text. 

We begin with the precise definition of the rate of private randomness distillation, both in the standard picture and in the dual one. In the standard picture \cite{BFW}, the protocol of extracting randomness from a bipartite state $\rho_{AE}$ is to perform a unitary map $U_{A\to K_AA'}$ followed by the measurement map
\be\label{measurement}
M_{A\to K_A}\otimes I_E(\cdot)=\sum_{i,j=1}^{|K_A||A'|}\ketbra{i}{i}_{K_A}\otimes\bra{ij}(\cdot)\ket{ij}_{K_AA'}
\ee
to approximate a state with $\log|K_A|$ randomness of the form $\frac{1}{|K_A|}\sum_{i=1}^{|K_A|}\ketbra{i}{i}_{K_A}\otimes \rho_E$, where ${\cal H}_A={\cal H}_{K_A}\otimes {\cal H}_{A'}$ and $\{i\}_{K_A}$ and $\{j\}_{A'}$ is the computational basis for subsystem $K_A$ and $A'$ respectively, and $\rho_E=\tr_{A}\rho_{AE}$. This is equivalent to performing the unitary $U_{A\to K_AA'}$ operation followed by the partial tracing of $A'$ and then a dephasing map on $K_A$ in the computational basis $\Delta_{K_A}(\cdot)=\sum_{i=1}^{|K_A|}\bra{i}(\cdot)\ket{i}_{K_A}\ketbra{i}{i}_{K_A}$. Notice that after the measurement map, the state of $K_A$ represents the classical outcome, so it cannot be involved in any further quantum processing, otherwise any initial state in $d-$dimensional space can always be used to generate $\log d$ private random bits by a projective measurement in an arbitrary basis followed by a further measurement in the complementary basis. Extending to the distributed scenario, the protocol of extracting independent randomness from a pure tripartite state $\ket{\psi}_{ABE}$ is: i) Alice performs a unitary $U_1: A\to A_1A_2$ and sends $A_2$ through the dephasing channel to Bob; ii) Bob performs a unitary $V_1: BA_2\to B_1B_2$ and sends $B_2$ through the dephasing channel to Alice; iii) Alice and Bob start the second round with the new state $\ket{\psi^{(2)}}_{A^{(2)}B^{(2)}E^{(2)}}=W_{B_2} \circ V_1\circ W_{A_2}\circ U_1\ket{\psi}_{ABE}$, where $W_X=\sum_{i=1}^{|X|}\ket{i}_{E_X}\ket{i}_X\bra{i}$ is the isometry of the dephasing channel on $X$ in the computational basis, and $A^{(2)}=A_1B_2$ and $B^{(2)}=B_1$, $E^{(2)}=EE_{A_2}E_{B_2}$; iv)  After some finite rounds, Alice and Bob perform local unitaries followed by the measurement map of the form (\ref{measurement}) to approximate a state with independent randomness $\log|K_A|$ at Alice's side and $\log|K_B|$ at Bob's side $\frac{1}{|K_A|}\sum_{i=1}^{|K_A|}\ketbra{i}{i}_{K_A}\otimes\frac{1}{|K_B|}\sum_{i=1}^{|K_B|}\ketbra{i}{i}_{K_B}\otimes \psi_{E_{f}}$, where $\psi_{E_{f}}$ is the final reduced state at Eve. Randomness generated in the intermediate measurement is required to be private against the final state of Eve, therefore the measurement can be postponed to the final stage of the protocol. So does for the intermediate partial trace. Denote the operation class composed by the four steps as MCLODCC and we are ready to define the rate of private randomness distillation in the standard picture. 

\begin{definition}
\label{distillable-rand-standard}
A protocol with $\{|K_{A}|, |K_{B}|,\epsilon(n)\}$, which distills private randomness from $n$ copies of $\ket{\psi}_{ABE}$ in the standard form, is an operation $\tilde{\Lambda}_n\in MCLODCC: A^nB^nE^n\to K_{A}K_{B}E^n_{f}$, where $|K_A|$ and $|K_B|$ are the sizes of the randomness parts dependent on $n$ and $\epsilon(n)$ is the error measured by the trace distance from the ideal standard form 
\[
\|\tilde{\Lambda}_n((\ketbra{\psi}{\psi}_{ABE}^{\otimes n})-\rho_{K_{A}K_{B}E^n_{f}}\|_1\le \epsilon(n),
\]
where 
\[
\rho_{K_{A}K_{B}E^n_{f}}\\
=\frac{1}{|K_{A}|}\sum_{i=1}^{|K_{A}|}\ketbra{i}{i}_{K_{A}}\otimes\frac{1}{|K_{B}|}\sum_{i=1}^{|K_{B}|}\ketbra{i}{i}_{K_{B}}\otimes \psi_{E^n_{f}}.
\]
The pair rate $(R_A, R_B)$ achieved by the protocol is defined as $R_A=\lim_{n\to\infty}\frac{\log |K_{A}|}{n}$ and $R_B=\lim_{n\to\infty}\frac{\log |K_{B}|}{n}$, when $\lim_{n\to\infty}\epsilon(n)= 0$.
\end{definition}

Similarly we define the pair rate of private randomness distillation in the dual picture where the goal state is the $\alpha$ state and the allowed operations come from the CLODCC class.

\begin{definition}
\label{distillable-rand-dual}
A protocol with $\{|K_{A}|, |K_{B}|,\epsilon(n)\}$, which distills private randomness from $n$ copies of $\ket{\psi}_{ABE}$ in the dual picture, is an operation $\Lambda_n\in CLODCC: A^nB^n\to K_{A}A'K_{B}B'$, where $|K_A|$ and $|K_B|$ are the sizes of the randomness parts and 
$\epsilon(n)$ is the distance from the $\alpha_{K_{A}A'K_{B}B'}$ state
\[
\|\Lambda_n(\psi_{AB}^{\otimes n})-\alpha_{K_{A}A'K_{B}B'}\|_1\le \epsilon(n),  
\]
and
\begin{align*}
&{\alpha}_{K_{A}A'K_{B}B'}\\
=& \frac{1}{|K_{A}||K_{B}|} \sum_{i,j=1}^{|K_{A}|} \sum_{k,\ell=1}^{|K_{B}|} 
                              \ketbra{i}{j}_{K_{A}}\otimes \ketbra{k}{\ell}_{K_{B}} 
                               \otimes U_{ik} \sigma_{A'B'} U_{j\ell}^{\dagger}.
\end{align*}
The pair rate $(R_A, R_B)$ achieved by the protocol is defined as $R_A=\lim_{n\to\infty}\frac{\log |K_{A}|}{n}$ and $R_B=\lim_{n\to\infty}\frac{\log |K_{B}|}{n}$, when $\lim_{n\to\infty}\epsilon(n)= 0$.
\end{definition}

 \medskip
 
%%%%%%%%%%%%%%%%%%%%%%%%
From Definitions \ref{distillable-rand-standard} \ref{distillable-rand-dual} and the proof of Lemma \ref{alpha-state}, we can establish a corollary that extracting private randomness in the two pictures are equivalent.
 \begin{corollary}
 \label{rand-equiv}
 If there exists a protocol to distill private randomness at the pair rate $\lim_{n\to\infty}(\frac{\log |K_{A}|}{n},\frac{\log |K_{B}|}{n})$, getting a   state that approximates the state $\rho_{K_{A}K_{B}E^n_{f}}$ under MCLODCC in the standard picture, then there exists a protocol achieving the same pair rate, getting a state that approximates ${\alpha}_{K_{A}A'K_{B}B'}$ under CLODCC in the dual picture. The converse statement is true as well.  
 \end{corollary}
%%%%%%%%%%%%%%%%%%%%%%%

\medskip
\begin{proofthm}{\bf of Lemma \ref{ibit-form}.}
In the usual approach, a quantum state having two independent 
random bits private against Eve is of the form
\be
  \label{product form}
  \rho_{K_AK_BE}=\frac{1}{4}\sum_{i,j=0}^{1} \ketbra{i}{i}_{K_A} \otimes \ketbra{j}{j}_{K_B} \otimes \rho_E.
\ee
Consider a pure state $|\phi\>_{K_AA'K_BB'E}$ whose marginal state $\phi_{K_AK_BE}$ turns 
into $\rho_{K_AK_BE}$ after measuring subsystems $K_A$ and $K_B$ in the computational basis. 
We denote the dephasing map in the computational basis by $\Delta$, so the 
dephased state of $\phi_{K_AK_BE}$ is 
\begin{align*}
  (\Delta_{K_A}\otimes\Delta_{K_B})\phi_{K_AK_BE}:=\\
 \sum_{ij} \ketbra{i}{i}_{K_A}\otimes \ketbra{j}{j}_{K_B}
                                                     \otimes \bra{ij}\phi_{K_AK_BE}\ket{ij}_{K_AK_B}.
\end{align*}
When written in the computational basis, the global state has the form
\be
\label{pure form}
  |\phi\>_{K_AA'K_BB'E}=\frac{1}{2}\sum_{i,j=0}^{1}\ket{i}_{K_A}\ket{j}_{K_B}\otimes \ket{\phi_{ij}}_{A'B'E}.
\ee
The dephased state is 
\[
  \rho_{K_AA'K_BB'E} 
   = \frac{1}{4}\sum_{i,j=0}^{1} \ketbra{i}{i}_{K_A}\otimes \ketbra{j}{j}_{K_B} 
                                 \otimes \ketbra{\phi_{ij}}{\phi_{ij}}_{A'B'E}.
\]
Because $\rho_{K_AK_BE}=\tr_{A'B'}\rho_{K_AA'K_BB'E}$ is of the form Eq. (\ref{product form}), 
$\ket{\phi_{ij}}_{A'B'E}$ satisfy the condition 
$\tr_{A'B'}\ketbra{\phi_{ij}}{\phi_{ij}}_{A'B'E}=\rho_E$ for 
all $i$ and $j$. Then, by the Schmidt decomposition, there exist unitary 
operators $U_{ij}$ acting on $A'B'$ such that the following holds:
\be
  |\phi_{ij}\>_{A'B'E}=U_{ij}|\phi_0\>_{A'B'E},
\ee
where $|\phi_0\>_{A'B'E}$ is a fixed purification for $\rho_E$. Using this 
form and tracing the E part in Eq. (\ref{pure form}), 
we get the state $\alpha_{K_AA'K_BB'}$ of the claimed structure. 
A similar reasoning gives the state $\alpha_{K_AA'B'}$ 
when random bit is only on one side. 

The robust version of Lemma \ref{ibit-form} is also true in the following sense,
which we state directly for general amounts of randomness for Alice and Bob:
Let, for positive integers $a$ and $b$,
\begin{align}
  \label{eq:random}
  &\rho_{K_AK_BE}      = \frac{1}{ab}\sum_{i=0}^{a-1}\sum_{j=0}^{b-1} 
                                    \ketbra{i}{i}_{K_A} \otimes \ketbra{j}{j}_{K_B} \otimes \rho_E, \\
  \label{eq:alpha}
  &\alpha_{K_AA'K_BB'} = \frac{1}{ab}\sum_{i,j=0}^{a-1}\sum_{k,\ell=0}^{b-1}
                                     \ketbra{i}{j}_{K_A} \otimes \ketbra{k}{\ell}_{K_B} 
                                      \otimes U_{ik} \sigma_{A'B'} U_{j\ell}^\dagger,
\end{align}
the ideal private randomness state and the ideal idit, respectively.
They descend from the common pure state
\be\label{pure}
  |\phi\>_{K_AA'K_BB'E} = \frac{1}{\sqrt{ab}}
                      \sum_{i=0}^{a-1}\sum_{j=0}^{b-1} \ket{i}_{K_A}\ket{j}_{K_B}
                                                       \otimes U_{ij}\ket{\phi_0}_{A'B'E},
\ee
with a suitable purification $\ket{\phi_0}$ of $\rho_E$, by which we mean 
$(\Delta_{K_A}\otimes\Delta_{K_B})\phi_{K_AK_BE}=\rho_{K_AK_BE}$, and $|\phi\>_{K_AA'K_BB'E}$ 
is a purification of $\alpha_{K_AA'K_BB'}$: $\tr_E \proj{\phi} = \alpha_{K_AA'K_BB'}$.

\begin{quote}
  If a ccq-state $\widetilde{\rho}_{K_AK_BE}$ is such that 
  $\frac12\|\widetilde{\rho}_{K_AK_BE} - \rho_{K_AK_BE}\|_1\leq\epsilon$,
  then for any pure state $\ket{\tilde\phi}_{K_AA'K_BB'E}$ satisfying 
  $(\Delta_{K_A}\otimes\Delta_{K_B})\tilde{\phi}_{K_AK_BE}=\tilde{\rho}_{K_AK_BE}$,
  \(
    \frac12\|\tilde\phi_{K_AA'K_BB'} - \alpha_{K_AA'K_BB'}\|_1 
              \leq \delta =\sqrt{f(\epsilon)(2-f(\epsilon))}
  \)
  with $f(\epsilon)\to 0$ as $\epsilon\to 0$, for an $\alpha$ of the form 
  (\ref{eq:alpha}) with suitable $\sigma_{A'B'}$ and unitaries $U_{ik}$ acting on $A'B'$.
  
  Conversely, if $\widetilde{\alpha}_{K_AA'K_BB'}$ is such that
  $\frac12\|\widetilde{\alpha}_{K_AA'K_BB'} - \alpha_{K_AA'K_BB'}\|_1\leq\epsilon$,
  then for any purification $\ket{\tilde\phi}_{K_AA'K_BB'E}$ of $\widetilde{\alpha}$,
  \(
    \frac12\|(\Delta_{K_A}\otimes\Delta_{K_B})\tilde\phi_{K_AK_BE} - \rho_{K_AK_BE}\|_1 
                                                \leq \delta=\sqrt{\epsilon(2-\epsilon)},
  \)
  for a $\rho$ of the form (\ref{eq:random}) with a suitable $\rho_E$.
\end{quote}

For the forward direction, suppose the ccq-state 
$\widetilde{\rho}_{K_AK_BE}=\sum_{ij}p_{ij}\ketbra{i}{i}_{K_A}\otimes \ketbra{j}{j}_{K_B}\otimes \tilde{\rho}_E^{ij}$ 
satisfies $\frac12\|\widetilde{\rho}_{K_AK_BE} - \rho_{K_AK_BE}\|_1\leq\epsilon$. 
We have a pure state 
\[
  \ket{\tilde{\phi}}_{K_AA'K_BB'E}=\sum_{ij}\sqrt{p_{ij}}\ket{i}_{K_A}\ket{j}_{K_B}\ket{\tilde{\phi}_{ij}}_{A'B'E'}
\] 
such that $(\Delta_{K_A}\otimes\Delta_{K_B})\tilde{\phi}_{K_AK_BE}=\widetilde{\rho}_{K_AK_BE}$, 
where $\ket{\tilde{\phi}_{ij}}_{A'B'E}$ are the purifications of $\tilde{\rho}_E^{ij}$. 
Denote a fixed purification of $\rho_E$ as $\ket{\phi_0}_{A'B'E}$. 
From Uhlmann's theorem \cite{Uhlmann}, we can choose unitaries $U_{ij}$ on $A'B'$ 
satisfying $\bra{\tilde{\phi}_{ij}}U_{ij}\ket{\phi_0}=F(\tilde{\rho}_E^{ij},\rho_E)$,
and form the state $\ket{\phi}_{K_AA'K_BB'E}$ in Eq. (\ref{pure}), which is clearly 
the purification for some $\alpha_{K_AA'K_BB'}$. We claim that $\ket{\tilde{\phi}}_{K_AA'K_BB'E}$ 
is close to $\ket{\phi}_{K_AA'K_BB'E}$ that is sufficient to prove $\tilde{\phi}_{K_AA'K_BB'}$ 
is close to $\alpha_{K_AA'K_BB'}$. 

By the triangle inequality for the trace norm and its monotonicity under 
partial trace, we get
\begin{align*}
\sum_{ij} &\,p_{ij}\|\tilde{\rho}_{E}^{ij}-\rho_E\|_1\\
&=\left\|\sum_{ij} p_{ij}\ketbra{ij}{ij}\otimes\tilde{\rho}_{E}^{ij}-\sum_{ij} p_{ij}\ketbra{ij}{ij}\otimes\rho_E\right\|_1\\
&\le\left\|\sum_{ij} p_{ij}\ketbra{ij}{ij}\otimes\tilde{\rho}_{E}^{ij}-\frac{1}{ab}\sum_{ij}\ketbra{ij}{ij}\otimes\rho_E\right\|_1\\
&\phantom{=}
 +\left\|\frac{1}{ab}\sum_{ij}\ketbra{ij}{ij}\otimes\rho_E-\sum_{ij} p_{ij}\ketbra{ij}{ij}\otimes\rho_{E}\right\|_1
\!\!\le 4\epsilon.
\end{align*}
From Lemma \ref{markov}, there exists a good index set 
$K=\left\{ij:\|\tilde{\rho}_E^{ij}-\rho_E\|_1\le\sqrt{4\epsilon}\right\}$, 
with $\sum_{ij\in K} p_{ij}\ge 1-\sqrt{4\epsilon}$. We denote $\bar{K}$ as the complement of $K$.
From Lemma \ref{FG}, we have $\bra{\tilde{\phi}_{ij}}U_{ij}\ket{\phi_0}\ge 1-\sqrt{\epsilon}$ for every $ij\in K$
and $\sum_{ij\in K}\sqrt{p_{ij}}\sqrt{\frac{1}{ab}}\ge 1-\epsilon$. Then we compute the fidelity,
\begin{align}
  \braket{\tilde{\phi}}{\phi}&_{K_AA'K_BB'E}
   =    \sum_{ij}\sqrt{p_{ij}}\sqrt{\frac{1}{ab}}\bra{\tilde{\phi}_{ij}}U_{ij}\ket{\phi_0} \\
  &\geq (1-\sqrt{\epsilon})\sum_{ij\in K}\sqrt{p_{ij}}\sqrt{\frac{1}{ab}}      \label{sum} \\
  &=    (1\!-\!\sqrt{\epsilon})\left(\sum_{ij} \!\sqrt{p_{ij}}\sqrt{\frac{1}{ab}}
                                 \!-\!\sum_{ij\in \bar{K}} \!\sqrt{p_{ij}}\sqrt{\frac{1}{ab}}\right) \\
  &\geq (1-\sqrt{\epsilon})\left[(1-\epsilon)
                          -\frac{1}{2}\sum_{ij\in \bar{K}}\left({p_{ij}}
                                                          +{\frac{1}{ab}}\right)\right] \label{square}\\
  &\geq (1\!-\!\sqrt{\epsilon})\left[(1\!-\!\epsilon)
              \!-\!\frac{1}{2}\sum_{ij\in \bar{K}}\left(\left|{p_{ij}}-{\frac{1}{ab}}\right|
                                                             +2p_{ij}\!\right)\!\right] \label{absvalue}\\
  &\geq (1-\sqrt{\epsilon})\left[(1-\epsilon)-\epsilon-\sqrt{4\epsilon}\right]          \label{partial} \\
  &=    (1-\sqrt{\epsilon})(1-2\epsilon-2\sqrt{\epsilon})=:1-f(\epsilon),
\end{align}
where Ineq. (\ref{sum}) comes from $\bra{\tilde{\phi}_{ij}}U_{ij}\ket{\phi_0}\ge 0$ 
and $\bra{\tilde{\phi}_{ij}}U_{ij}\ket{\phi_0}\ge 1-\sqrt{\epsilon}$ for $ij\in K$, (\ref{square}) from $\sqrt{x}\sqrt{y}\le\frac12(x+y)$, (\ref{absvalue}) from $x-y\le|x-y|$, and (\ref{partial}) from $\frac{1}{2}\sum_{ij\in \bar{K}}|{p_{ij}}-{\frac{1}{ab}}|\le\frac{1}{2}\sum_{ij}|{p_{ij}}-{\frac{1}{ab}}|\le\epsilon$ and $\sum_{ij\in K} p_{ij}\ge 1-\sqrt{4\epsilon}$. 
Then we get, by monotonicity of the fidelity, that
\[
  F(\tilde{\phi}_{K_AA'K_BB'},\alpha_{K_AA'K_BB'}) \ge 1-f(\epsilon),
\]
and so
\[
  \frac12 \|\tilde{\phi}_{K_AA'K_BB'}-\alpha_{K_AA'K_BB'}\|_1 
                             \leq \sqrt{f(\epsilon)(2-f(\epsilon))} =: \delta.
\]
                             
For the opposite direction, we proceed in a simpler way,
namely by Uhlmann's theorem we have that there exists a unitary $V_E$
such that
\[
  \bigl|\bra{\tilde{\phi}}\1\otimes V\ket{\phi}\bigr| 
      = F(\widetilde{\alpha}_{K_AA'K_BB'},\alpha_{K_AA'K_BB'})\geq 1-\epsilon.
\]
Then we get, by monotonicity of the fidelity, that
\[\begin{split}
&F\bigl((\Delta_{K_A}\otimes\Delta_{K_B})\tilde{\phi}_{K_AK_BE},
       (\Delta_{K_A}\otimes\Delta_{K_B})(V{\phi}_{K_AK_BE}V^{\dagger})\bigr) \\
      &\ge 1-\epsilon,
\end{split}\]
and so
\[\begin{split}
  &\frac12 \|(\Delta_{K_A}\otimes\Delta_{K_B})\tilde{\phi}_{K_AK_BE}
             -(\Delta_{K_A}\otimes\Delta_{K_B})(V{\phi}_{K_AK_BE}V^{\dagger})\|_1 \\
                                            &\leq \sqrt{\epsilon(2-\epsilon)} =: \delta.
\end{split}\]
As before, $(\Delta_{K_A}\otimes\Delta_{K_B})(V{\phi}_{K_AK_BE}V^{\dagger})$ has 
the same structure as $\rho_{K_AK_BE}$.
\end{proofthm}

%%%%%%%%%%%%%%%%%%%%%%%
\medskip
{\theorem \label{two-decouple} 
Given $\epsilon,\delta >0$, for $n$
copies of a pure tripartite state $\ket{\psi}_{ABE}$ where $n$ is large, there exists a 
unitary $U:A^n\to KA'$ with a fixed basis $\{\ket{i}\}$ 
of subsystem $K$, $(U_{A^n}\otimes \1_{B^nE^n})\ket{\psi}_{ABE}^{\otimes n}
   = \sum_{i=1}^{|K|} \sqrt{p_i} \ket{i}_K\ket{\psi_i}_{A'B^nE^n}$, such that after measurement on $K$ in the fixed basis,
\begin{align*}
 \left\| \sum_{i=1}^{|K|}p_i\ketbra{i}{i}_{K}\otimes\psi^i_{B^n}
          - \frac{1}{|K|}\sum_{i=1}^{|K|} \ketbra{i}{i}_{K}\otimes \psi_{B}^{\otimes n} \right\|_1
  \le \epsilon, \\
  \left\| \sum_{i=1}^{|K|}p_i\ketbra{i}{i}_{K}\otimes\psi^i_{E^n}
          - \frac{1}{|K|}\sum_{i=1}^{|K|}\ketbra{i}{i}_{K}\otimes\psi_{E}^{\otimes n} \right\|_1
  \le \epsilon,
\end{align*}
when $\frac{1}{n} \log|K| = \min\{I(A:E)_{\psi}, I(A:B)_{\psi}\}-\delta$.
}
\medskip

\begin{proof}
First we prove a one-shot version of Theorem \ref{two-decouple}.
Namely, for a pure tripartite state $|\psi\>_{ABE}$, there exists a unitary 
$U:A\to KA'$ with  a fixed standard basis $\{\ket{i}\}$ of subsystem $K$, 
such that the state of $K$ after the von Neumann measurement simultaneously 
decouples from Bob and Eve, as follows:
\begin{equation}\begin{split}
  \label{AE}
  &\left\| \sum_{i=1}^{|K|}p_i\ketbra{i}{i}_{K}\otimes\psi^i_{B}
            - \frac{1}{|K|}\sum_{i=1}^{|K|}\ketbra{i}{i}_{K}\otimes\psi_{B} \right\|_1 \\
  &\phantom{=====}
   \le 2(2^{-\frac{1}{2}[H^{\epsilon}_{\min}(A|B)_{\psi}+H^{\epsilon}_{\min}(R|K)_{\tau}]}+12\epsilon),
\end{split}\end{equation}
\begin{equation}\begin{split}
  \label{AB}
  &\left\| \sum_{i=1}^{|K|}p_i\ketbra{i}{i}_{K}\otimes\psi^i_{E}
            - \frac{1}{|K|}\sum_{i=1}^{|K|}\ketbra{i}{i}_{K}\otimes\psi_{E} \right\|_1 \\
  &\phantom{=====}
   \le 2(2^{-\frac{1}{2}[H^{\epsilon}_{\min}(A|E)_{\psi}+H^{\epsilon}_{\min}(R|K)_{\tau}]}+12\epsilon), \end{split}\end{equation}
where $(U_A\otimes \1_{BE})\ket{\psi}_{ABE}=\sum_{i=1}^{|K|}\sqrt{p_i}\ket{i}_K\ket{\psi_i}_{A'BE}$
and $\tau_{RK}=(\id_R\otimes T)\Phi_{RA}$ is the Choi-Jamio\l{}kowski state of the 
CPTP map
\[
  T_{A\to K}(\rho)=\sum_i\ketbra{i}{i}(\tr_{A'}\rho)\ketbra{i}{i},
\]
with $A=A'K$, that is tracing $A'$ first and then performing a von 
Neumann measurement on $K$. Notice that $\tau_K=\frac{\1_K}{|K|}$. 
Now borrowing the idea from \cite{Ye}, and Lemma \ref{Dupuis}, it is 
easy to see that more than half of the unitary operators $U$ satisfy 
Ineq. (\ref{AE}) and more than half satisfy Ineq. (\ref{AB}). 
So there must exist a common unitary such that both inequalities are satisfied.  

Now, by the asymptotic equipartition property, the $\epsilon$-smooth 
conditional min-entropy reduces to conditional entropy in the asymptotic 
i.i.d. setting, and
$S(A|E)_{\psi^{\otimes n}} = n(S(B)-S(E))$, 
$S(A|B)_{\psi^{\otimes n}}=n(S(E)-S(B))$, as well as
$S(R|K)_{\tau}=nS(A)-\log|K|$ when we restrict to the typical subspace of 
$\psi_A^{\otimes n}$. Counting the rates, we get that if 
$\frac{1}{n}\log|K| = \min\{I(A:B),I(A:E)\}-\delta$, then the first term 
will decay exponentially with $n$. So we can get randomness in $K$ simultaneously 
independent of Bob's and of Eve's system. 
\end{proof}

\medskip\noindent
\textbf{Remark}\ \,Theorem \ref{two-decouple} can be 
generalized to multipartite states.

%%%%%%%%%%%%%%%%%%%

{\theorem \label{security}
Given $\epsilon, \delta>0$ and sufficiently large $n,m$, for a state 
$\ket{\psi}_{ABE}$ satisfying $S(A|B)>0>S(B|A)$, there exist local 
unitaries $U:A^{n+m}\to K_AA'$ and $V:B^{n+m}\to K_BB'$, 
and fixed bases $\{\ket{i}\}_{K_A}$  and $\{\ket{j}\}_{K_B}$,
\[
  U\otimes V\ket{\psi}_{ABE}^{\otimes (n+m)}
   =\! 
   \sum_{i,j=1}^{|K_A|,|K_B|}\!\! \sqrt{p_{ij}} \ket{i}_{K_A} \ket{j}_{K_B}\ket{\psi_{ij}}_{A'B'E^nE^m},
\]
so that after measurements on $K_A, ~K_B$ in the fixed bases,
\begin{align*}
 &\left\| \sum_{i,j} \! p_{ij}\ketbra{ij}{ij}_{K_AK_B}\!\otimes\!\psi^{ij}_{E^nE^m}
          \!\!-\! \frac{\1_{K_A}}{|K_A|}\!\otimes\!\frac{\1_{K_B}}{|K_B|}
                \!\otimes\!\psi_{E}^{\otimes (n+m)}\right\|_1 \!\!\! \\
    \le &f'(\epsilon):=2\epsilon+5\sqrt{2\epsilon}+2\sqrt{4\sqrt{2\epsilon}-2\epsilon}, 
\end{align*}
%\ben
%&&\left\| \sum_{i,j=1}^{|K_A|,|K_B|}p_{ij}\ketbra{ij}{ij}_{K_AK_B}\otimes\psi_{ij}^{E^nE^m}\right.\nonumber \\ 
%&&~    - \frac{\1_A}{|K_A|}\left. \otimes \frac{\1_B}{|K_B|} \otimes\psi_{E}^{\otimes (n+m)} \right\|_1 \le \epsilon, \nonumber
%\een
assuming $\log|K_A|=nI(A:B)-n\delta$ and $\log|K_B|=\min\{mI(B:A),mI(B:E)+2nS(B)\}-m\delta$. By choosing $\frac{m}{n}=-\frac{S(B|A)}{S(B)}$, $\log|K_B|=mI(A:B)-m\delta$.
}

\medskip
Here $S(A|B)>0>S(B|A)$ implies that $I(A:E)>I(A:B)>I(B:E)$. From Theorem \ref{two-decouple}, we know that Alice can distill private randomness at the rate $I(A:B)$ but Bob cannot achieve the rate $I(A:B)$ directly. In the proof of Theorem \ref{two-side}, we will see that $I(A:B)$ is the optimal rate at which private randomness can be distilled from the information part. To that end, Bob will make use of the local noise created from Alice's randomness distillation. 

\medskip

\begin{proof}
Essentially Theorem \ref{security} says that Theorem \ref{two-decouple} is 
composable in the following way. Suppose $n$ copies of $\ket{\psi}_{ABE}$ 
and $m$ copies of $\ket{\psi}_{ABE}$ are given. 
From Theorem \ref{two-decouple}, we know that Alice can get $nI(A:B)_{\psi}$
bits of randomness in $K_A$ from 
$\ket{\psi}_{ABE}^{\otimes n}$. Notice that Bob could virtually regard his 
state $\psi_B^{\otimes n}$ as local noise in the sense that it is entangled 
with $A'E$ even assuming $A'$ is at Eve's control which is the worst case 
(but of course Eve does not really have access to $A'$ whose state may 
have information on the randomness). And he can make use of this noise 
to help extract randomness with the fresh state $\ket{\psi}_{ABE}^{\otimes m}$, 
where the randomness in subsystem $K_B$ will decouple from the system $A'E$. 
All together, the randomness in $K_A$ and $K_B$ is independent and secure 
against Eve. 

The first step is for Alice to extract randomness from $\ket{\psi}_{ABE}^{\otimes n}$. 
From Theorem \ref{two-decouple}, we know that there exists 
a unitary $U:A^n\to K_AA'$ such that 
\ben
U\ket{\psi}_{ABE}^{\otimes n}
   = \sum_{i=1}^{|K_A|} \sqrt{p_{i}} \ket{i}_{K_A} \ket{\psi_{i}}_{A'B^nE^n},
\een
satisfying 
\ben
\left\| \sum_{i=1}^{|K_A|}p_i\ketbra{i}{i}_{K_A}\otimes\psi^i_{B^n}
          - \frac{1}{|K_A|}\sum_{i=1}^{|K_A|} \ketbra{i}{i}_{K_A}\otimes \psi_{B}^{\otimes n} \right\|_1
  \le \epsilon,\nonumber\\ 
  \left\| \sum_{i=1}^{|K_A|}p_i\ketbra{i}{i}_{K_A}\otimes\psi^i_{E^n}
          - \frac{1}{|K_A|}\sum_{i=1}^{|K_A|}\ketbra{i}{i}_{K_A}\otimes\psi_{E}^{\otimes n} \right\|_1
  \le \epsilon, \nonumber
\een
when $\log|K_A|=n[I(A:B)-\delta]$. 
By the monotonicity of trace norm, we have 
\[
 \sum_{i=1}^{|K_A|}p_i\ketbra{i}{i}_{K_A}\stackrel{\epsilon}{\approx}
          \frac{1}{|K_A|}\sum_{i=1}^{|K_A|} \ketbra{i}{i}_{K_A} .
\]
By triangle inequality of the trace norm, we then get
\[
 \frac{1}{|K_A|}\sum_{i=1}^{|K_A|} \ketbra{i}{i}_{K_A}\otimes \psi^i_{B^n}\stackrel{2\epsilon}{\approx}
 \frac{1}{|K_A|}\sum_{i=1}^{|K_A|} \ketbra{i}{i}_{K_A}\otimes \psi_{B}^{\otimes n}.
\]
From Lemma \ref{markov}, we have 
\be
  \sum_{i: \psi^i_{B^n}\stackrel{\sqrt{2\epsilon}}{\approx}\psi_{B}^{\otimes n}}
                                               \frac{1}{|K_A|}\ge 1-\sqrt{2\epsilon}.
\ee
Similarly we have
\be
  \sum_{i: \psi^i_{E^n}\stackrel{\sqrt{2\epsilon}}{\approx}\psi_{E}^{\otimes n}}
                                               \frac{1}{|K_A|}\ge 1-\sqrt{2\epsilon}.
\ee
For both to be true, we have
\be
\sum_{\widetilde{K}_A}\frac{1}{|K_A|}\ge 1-2\sqrt{2\epsilon},
\ee
where 
$\widetilde{K}_A:=\left\{i: \psi^i_{B^n}\stackrel{\sqrt{2\epsilon}}{\approx}\psi_{B}^{\otimes n}
                        \ \&\ 
                        \psi^i_{E^n}\stackrel{\sqrt{2\epsilon}}{\approx}\psi_{E}^{\otimes n}\right\}$.
Denoting $R_A=[K_A]\setminus\widetilde{K}_A$, we can write
\begin{align*}
\frac{1}{|K_A|}\sum_{i=1}^{|K_A|} \ketbra{i}{i}_{K_A}&\otimes \psi^i_{A'B^nE^n}       \\
 &= \frac{1}{|K_A|}\sum_{i\in\tilde{K}_A}\ketbra{i}{i}_{K_A}\otimes \psi^i_{A'B^nE^n} \\
 &\phantom{=}
    + \frac{1}{|K_A|}\sum_{i\in R_A}\ketbra{i}{i}_{K_A}\otimes \psi^i_{A'B^nE^n}.
\end{align*}

The next step is for Bob to extract randomness by applying a unitary 
$V: B^{(n+m)}\to K_BB'$ to the state 
$\ket{\phi}_{B^n(A'E^{n})}\otimes\ket{\psi}_{ABE}^{\otimes m}$, 
where $A'E^n$ is virtually at Eve's side and $\ket{\phi}_{B^n(A'E^{n})}$ 
is any state satisfying $\phi_{B^n}=\psi_{B}^{\otimes n}$. We assume a unitary $W: A'E^n\to E_v^n$ such that $W\ket{\phi}_{B^n(A'E^{n})}=\ket{\phi}_{BE_{v}}^{\otimes n}$. From Theorem \ref{two-decouple}, we have
\[
V\ket{\phi}_{BE_v}^{\otimes n}\otimes\ket{\psi}_{ABE}^{\otimes m}=\sum_{j=1}^{|K_B|}\sqrt{q_j}\ket{j}_{K_B}\ket{\phi_j}_{B'A^mE_v^nE^m},
\]
satisfying
\begin{align}
%\label{}
&\left\| \sum_{j=1}^{|K_B|}q_j\ketbra{j}{j}_{K_B}\!\otimes\phi^j_{A^m}
         \!-\! \frac{1}{|K_B|}\sum_{j=1}^{|K_B|}\ketbra{j}{j}_{K_B}\!\otimes\psi_{A}^{\otimes m} \right\|_1 
         \!\!\!\le \epsilon,\\ 
&\left\| \sum_{j=1}^{|K_B|}q_j\ketbra{j}{j}_{K_B}\otimes\phi^j_{E_v^nE^m}\right. \nonumber\\ 
&\phantom{=====}
 \left. -\sum_{j=1}^{|K_B|}\!\frac{1}{|K_B|}\ketbra{j}{j}_{K_B}
                                     \!\otimes\!\phi_{E_v}^{\otimes n}\!\otimes\!\psi_{E}^{\otimes n} \right\|_1 
  \!\!\le \epsilon, \label{B-decouple-E}
\end{align}
when $\log|K_B|=\min\{mI(B:A), mI(B:E)+2nS(B)\}-m\delta$. Here the additional term $2nS(B)$ comes from the mutual information between $B^n$ and $E_v^n$ in the state of noise $\ket{\phi}_{BE_v}^{\otimes n}$. Applying $W^{\dagger}: E_v^n\to A'E^n$ to reverse the action of $W$ does not change the above relations but recover states on $A'E^n$. Thus Theorem \ref{two-decouple} can be applied though $\ket{\phi}_{B^n(A'E^{n})}$ is not of product form. Notice that $V$ is determined entirely by Bob's local states $\phi_{B^n}=\psi_B^{\otimes n}$ and 
$\ket{\psi}_{ABE}^{\otimes m}$, and is independent of the local state $\phi_{A'E^{n}}$. 
Thus Bob can apply $V$ blindly to all states of the form 
$\ket{\psi_i}_{A'B^nE^n}\otimes \ket{\psi}_{ABE}^{\otimes m}$. 

Suppose now that
\[
V\ket{\psi_i}_{B^n(A'E^{n})}\otimes\ket{\psi}_{ABE}^{\otimes m}=\sum_{j=1}^{|K_B|}\sqrt{q_{j|i}}\ket{j}_{K_B}\ket{\psi_{j|i}}_{B'A^mA'E^nE^m}.
\]
For $i\in\tilde{K}_A$, we have 
$\psi^i_{B^n}\stackrel{\sqrt{2\epsilon}}{\approx}\psi_{B}^{\otimes n}=\phi_{B^n}$. 
Using Uhlmann's theorem (precisely, first using Lemma \ref{FG}, then Lemma \ref{Uhlmann}),  
there exist unitary operators $W_i$ acting on $A'E^n$ such that the purification 
$W_i(\phi_{B^n(A'E^n)})W_i^{\dagger}=\phi^i_{B^n(A'E^n)}\stackrel{\epsilon'}{\approx}\psi^i_{B^n(A'E^n)}$ 
with $\phi^i_{B^n}=\psi_B^{\otimes n}$ and 
$\phi^i_{A'E^n}\stackrel{\epsilon'}{\approx}\psi^i_{A'E^n}$, where 
$\epsilon'=\sqrt{4\sqrt{2\epsilon}-2\epsilon}$. Then we have
\be
\label{approx}
V\ket{\psi_i}_{B^n(A'E^{n})}\otimes\ket{\psi}_{ABE}^{\otimes m}\stackrel{\epsilon'}{\approx}V\otimes W_i\ket{\phi}_{B^n(A'E^{n})}\otimes\ket{\psi}_{ABE}^{\otimes m}.
\ee
From Eq. (\ref{approx}), we get
\begin{align}
\sum_{j}q_{j|i}&\ketbra{j}{j}_{K_B}\otimes \psi^{j|i}_{A'E^{n}E^{m}} \nonumber\\
 &\stackrel{\epsilon'}{\approx} \sum_{j}q_j\ketbra{j}{j}_{K_B}\otimes W_i\phi^{j}_{A'E^{n}E^{m}}W^{\dagger}_i \label{one}\\
 &\stackrel{\epsilon}{\approx} \frac{1}{|K_B|}\sum_{j}\ketbra{j}{j}_{K_B}\otimes W_i\phi_{A'E^{n}}W^{\dagger}_i\otimes\psi_{E}^{\otimes m} \label{two}\\
 &\stackrel{\epsilon'}{\approx} \frac{1}{|K_B|}\sum_{j}\ketbra{j}{j}_{K_B}\otimes \psi^{i}_{A'E^{n}}\otimes\psi_{E}^{\otimes m}, \label{three}
\end{align}
where Eq. (\ref{one}) comes from Eq. (\ref{approx}) and the monotonicity of trace norm 
under the dephasing operation on $K_B$ and tracing out $A^m$, 
Eq. (\ref{two}) from (\ref{B-decouple-E}), and Eq. (\ref{three}) from 
$W_i\phi_{A'E^{n}}W^{\dagger}_i=\phi^i_{A'E^n}\stackrel{\epsilon'}{\approx}\psi^i_{A'E^n}$.
Summing the errors, we have for $i\in\tilde{K}_A$
\begin{align}
\sum_{j}q_{j|i}&\ketbra{j}{j}_{K_B}\otimes \psi^{j|i}_{A'E^{n}E^{m}}\nonumber\\
& \stackrel{2\epsilon'+\epsilon}{\approx} \frac{1}{|K_B|}\sum_{j}\ketbra{j}{j}_{K_B}\otimes \psi^{i}_{A'E^{n}}\otimes\psi_{E}^{\otimes m}.
\end{align}

For $i\in\tilde{K}_A$, we also have 
$\psi^i_{E^n}\stackrel{\sqrt{2\epsilon}}{\approx}\psi_{E}^{\otimes n}$, 
and by monotonicity of the trace norm under tracing out $A'$,
\begin{align}
\sum_{j}q_{j|i}&\ketbra{j}{j}_{K_B}\otimes \psi^{j|i}_{E^{n}E^{m}} \nonumber\\
&\stackrel{2\epsilon'+\epsilon}{\approx}\frac{1}{|K_B|}\sum_{j}\ketbra{j}{j}_{K_B}\otimes \psi^{i}_{E^{n}}\otimes\psi_{E}^{\otimes m}\\
&\stackrel{\sqrt{2\epsilon}}{\approx} \frac{1}{|K_B|}\sum_{j}\ketbra{j}{j}_{K_B}\otimes \psi_{E}^{\otimes n}\otimes\psi_{E}^{\otimes m}\\
&=\frac{1}{|K_B|}\sum_{j}\ketbra{j}{j}_{K_B}\otimes \psi_{E}^{\otimes (m+n)}.
\end{align}
Summing the errors, we get
\begin{align}
\sum_{j}q_{j|i}&\ketbra{j}{j}_{K_B}\otimes \psi^{j|i}_{E^{n}E^{m}} \nonumber\\
&\stackrel{\epsilon''}{\approx}\frac{1}{|K_B|}\sum_{j}\ketbra{j}{j}_{K_B}\otimes \psi_{E}^{\otimes (m+n)},
\end{align}
where $\epsilon'':=2\epsilon'+\epsilon+\sqrt{2\epsilon}$.

Now combining the two steps and considering the dephasing state, we get
\begin{align}
\sum_{i\in {K}_{A}}&{p}_i \ketbra{i}{i}_{{K}_{A}} \otimes 
        \sum_{j\in {K}_B} {{q}_{j|i}} \ketbra{j}{j}_{{K}_B}\otimes\psi^{j|i}_{E^{n}E^{m}} \nonumber\\
&\stackrel{\epsilon}{\approx}  
 \frac{1}{|{K}_{A}|}\sum_{i\in {K}_{A}} \!\ketbra{i}{i}_{{K}_{A}}\!\otimes\!
             \sum_{j\in {K}_B} \!{{q}_{j|i}} \ketbra{j}{j}_{{K}_B}\!\otimes\!\psi^{j|i}_{E^{n}E^{m}} \\
&=
 \frac{1}{|{K}_{A}|}\sum_{i\in \tilde{K}_{A}} \!\ketbra{i}{i}_{{K}_{A}}\!\otimes\!
           \sum_{j\in {K}_B} \!{{q}_{j|i}} \ketbra{j}{j}_{{K}_B}\!\otimes\!\psi^{j|i}_{E^{n}E^{m}}\nonumber\\
&\phantom{=}
 +\frac{1}{|{K}_{A}|}\sum_{i\in {R}_{A}} \!\ketbra{i}{i}_{{K}_{A}}\!\otimes\!
    \sum_{j\in {K}_B} {{q}_{j|i}} \!\ketbra{j}{j}_{{K}_B}\!\otimes\!\psi^{j|i}_{E^{n}E^{m}} \\
&\stackrel{\epsilon''}{\approx} 
 \frac{1}{|{K}_{A}|}\sum_{i\in \tilde{K}_{A}} \!\ketbra{i}{i}_{{K}_{A}}\!\otimes\!
 \frac{1}{|K_B|}\sum_{j} \!\ketbra{j}{j}_{K_B}\!\otimes\!\psi_{E}^{\otimes (n+m)} \nonumber\\
&\phantom{=}
 +\frac{1}{|{K}_{A}|}\sum_{i\in {R}_{A}}\!\ketbra{i}{i}_{{K}_{A}}\!\otimes\!
  \sum_{j\in {K}_B} {{q}_{j|i}}\!\ketbra{j}{j}_{{K}_B}\!\otimes\!\psi^{j|i}_{E^{n}E^{m}} \\
&\stackrel{4\sqrt{2\epsilon}}{\approx}\!
 \frac{1}{|{K}_{A}|}\sum_{i\in {K}_{A}} \!\ketbra{i}{i}_{{K}_{A}}\!\otimes\!
 \frac{1}{|K_B|}\sum_{j}\!\ketbra{j}{j}_{K_B}\!\otimes\!\psi_{E}^{\otimes (n+m)},\label{bad}
\end{align}
where Eq. (\ref{bad}) comes from $\sum_{i\in {R}_{A}}\frac{1}{|{K}_{A}|}\le 2\sqrt{2\epsilon}$ 
and $\|\rho_1-\rho_2\|_1 \le 2$ for any two states.
We sum up all the errors and get
\begin{align}
&\left\| \sum_{i\in {K}_{A}} {{p}_i}\ketbra{i}{i}_{{K}_{A}}\otimes 
         \sum_{j\in {K}_B} {{q}_{j|i}}\ketbra{j}{j}_{{K}_B}\otimes\psi^{j|i}_{E^{n}E^{m}}\right.\nonumber\\ 
&\phantom{=}
 \left. -\frac{1}{|K_A|}\sum_{i\in K_A} \!\ketbra{i}{i}_{K_A}\!\otimes\! 
         \frac{1}{|K_B|}\sum_{j\in K_B} \!\ketbra{j}{j}_{K_B}\!\otimes\!\psi_{E}^{\otimes (n+m)} \right\|_1  \nonumber\\
&\phantom{====}
 \leq  2\epsilon+5\sqrt{2\epsilon}+2\sqrt{4\sqrt{2\epsilon}-2\epsilon}=:f'(\epsilon).
\end{align}
The protocol can be composed iteratively in a new round by tackling the indices of $K_AK_B$ together. Notice that only the existence of the good index set is required in the proof, which is used to construct the ideal control state. Alice and Bob need not know how to select the good index set because the unitary operation, as explained by argument below Eq. (\ref{B-decouple-E}), depends only on the average local state. Thus they can apply the decoupling unitary locally such that the dephasing state on the randomness subsystem will approximate the ideal randomness state. Indeed iterating the protocol in finite rounds renders the security of the ``ping-pong'' protocol in the proof of Theorem \ref{two-side}.
\end{proof}

%%%%%%%%%%%%%%%%%%%%%%%%
\medskip
\begin{proofthm}{\bf of Theorem \ref{two-side}.}
We first prove it for the setting 1) of no noise and no communication. 
The other three settings can be reduced to setting 1), making
a couple of simple observations.

The proof for setting 1) proceeds by the analysis of four cases. A common feature is that Alice and Bob preprocess their states by Lemma \ref{preprocess}, separating the purity parts that can be exchanged when dephasing channels are allowed, and the information parts that contain 
the correlation. The private randomness contributed from purity part is equal to the purity. Then we only need to deal with the information part and apply iteratively Theorem \ref{two-decouple} for the special state $\psi_{A_IB_IE}$ satisfying the conditions $S(B_I)=\log|B_I|$ and $S(A_I)=\log|A_I|$. When $I(A:B)\le I(A:E)$ or $I(A:B)\le I(B:E)$, we can distill private randomness at the rate $I(A:B)$ without need of local noise. Otherwise, we will employ or create local noise such that the rate is still equal to $I(A:B)$ by noticing that local noise implies increasing $I(A:E)$ or $I(B:E)$ but keeping $I(A:B)$ invariant. Then we need Theorem \ref{security} to guarantee that the randomness generated at both sides is independent.

\emph{Case 1. $S(A|E)<0$ and $S(B|E)< 0$}: By the duality relation between conditional entropies of a pure tripartite state $\ket{\psi}_{ABE}$, that is $S(A|E)+S(A|B)=0$ and $S(B|E)+S(B|A)=0$, 
The claimed region is
\begin{align}
  R_A       &\leq \log|A| - S(A|B), \nonumber\\
  R_B       &\leq \log|B| - S(B|A), \\
  R_A + R_B &\leq R_G. \nonumber
\end{align}
The extremal points are given by the rate pairs $(R_A,R_B)=(\log|A|+S(B)-S(AB),\log|B|-S(B))$
and $(R_A,R_B)=(\log|A|-S(A),\log|B|+S(A)-S(AB))$. The first point is simply achieved by Alice extracting randomness from her
local purity part at rate $\log|A|-S(A)$ and from the information part $I(A:B)$ by Theorem \ref{two-decouple}, and Bob 
just extracting randomness from the local purity part at rate $\log|B|-S(B)$. Exchanging the role gives the second point.   
The rate region is achieved by time-sharing. Time-sharing is a technique used in the asymptotic iid setting in information theory. By sharing time in a sense that for $np$ out of $n$ runs, protocol 1 is performed, and for $n(1-p)$ out of $n$, protocol 2 is performed, if the rates $R_1$ and $R_2$ are achievable in protocol 1 and 2 respectively, then the rate $R=pR_1+(1- p)R_2$ with $0<p<1$ is also achievable \cite[p.534]{Cover-Thomas}.

\emph{Case 2. $S(A|E)< 0$ and $S(B|E)\geq 0$}: By the duality relation, the claimed region is 
\begin{align}
  R_A       &\leq \log|A| - S(A|B), \nonumber\\
  R_B       &\leq \log|B|, \\
  R_A + R_B &\leq R_G. \nonumber
\end{align}
One of the extremal points is $(\log|A| - S(A|B), \log|B|-S(B))$ by the same 
protocol as above. The other extremal point is $(\log|A|-S(AB),\log|B|)$, which
is achieved as follows:
Alice distills local purity from all copies, then distills randomness from $pn$ copies 
of the informational part by Theorem \ref{two-decouple}. This produces 
$R_A=\log|A|-S(A)+pI(A:B)$ of randomness on Alice's side, and at the same time creates local noise $\psi^{\otimes pn}_{(A'E):B}$ for Bob; Then Bob extracts randomness from the state $\psi^{\otimes (1-p)n}_{ABE}\otimes \psi^{\otimes pn}_{(A'E):B}=:\psi'$. The randomness comes from local purity for all copies that contribute the rate $\log|B|-S(B)$, and the informational part for $(1-p)n$ copies under the assistance of local noise $\psi^{\otimes pn}_{(A'E):B}$. By Theorem \ref{two-decouple}, if we choose $p$ satisfying the
$I(B:E)_{\psi'}=I(B:A)_{\psi'}$, i.e. $(1-p)S(B|E)-pS(B)=0$ then it will produce $(1-p)I(A:B)$ of randomness from the informational part. The sum of  the two parts gives $\log|B|$ on Bob. Replacing $p$ in $R_A$ then gives $R_A=\log|A|-S(AB)$. That the randomness at both sides is independent and private against Eve comes from Theorem \ref{security}.

\emph{Case 3. $S(A|E)\ge 0$ and $S(B|E)< 0$} is analogous to case 2.

\emph{Case 4. $S(A|E)> 0$ and $S(B|E)> 0$}: The claimed region is
\begin{align}
  R_A       &\leq \log|A|, \nonumber\\
  R_B       &\leq \log|B|, \\
  R_A + R_B &\leq R_G. \nonumber
\end{align}
Its extremal points are $(R_A,R_B)=(\log|A|,\log|B|-S(AB))$
and $(R_A,R_B)=(\log|A|-S(AB),\log|B|)$. 
We suppose Alice and Bob preprocess their states by local compression,
separating the purity parts and the information parts that contain 
the correlation. We only deal with the information part. Neither Alice nor Bob 
can directly distill randomness at the rate $I(A:B)$. We will show the ``ping-pong'' protocol achieves the two extremal points. 
Alice and Bob will take a ``ping-pong'' strategy 
to help each other obtain required virtual local noise between Alice and Eve, and Bob 
and Eve in order to distill randomness at the rate $I(A:B)$. It is an activation process and it works only 
if the process could be amplified to more and more copies. So it 
has to satisfies some constraints. However finally we find that the constraint 
is always satisfied. So there is no constraint at all!   

In detail, the ``ping-pong'' protocol works as follows: suppose that initially Bob puts $x$ of his copies into his shield system (the copies 
consumed will be negligible when we count the rate at the end) and will not use 
them any more. Alice can imagine Bob's $x$ part is in Eve's hands, i.e. $\psi_{A:(BE)}^{\otimes x}$, so she would 
share $xS(A)$ ebits with Eve by her $x$ systems. Then Alice distills randomness on 
$n_1$ fresh copies using $xS(A)=1$ ebits. Now we apply Theorem \ref{two-decouple} to the state $\psi_{ABE}^{\otimes n_1}\otimes \psi_{A:(BE)}^{\otimes x}$. If $n_1$ satisfies $n_1I(A:B)=n_1I(A:E)+2$, i.e. $n_1S(A|E)=1$, then Alice can distill randomness at the rate $I(A:B)$. Notice that the rate would be $I(A:E)$ if no local noise is used, which is less than $I(A:B)$ from the condition $S(A|E)>0$. In this sense we say that   
one ebit can boost randomness distillation of $n_1$ copies at the rate $I(A:B)$. 
Alice's randomness distillation on $n_1$  copies creates new virtual ebits between Bob and Eve, the number of which is 
$n_1S(B)$. Then Bob distills randomness on $n_2$ fresh copies under 
the assistance of $n_1S(B)$ ebits, where $n_2$ satisfies $n_2S(B|E)=n_1S(B)$. 
After Bob's action, Alice would virtually share $n_2S(A)$ ebits with Eve. This
is one round. The process can be amplified if $n_2S(A)>1$, which reads $S(A)S(B)>S(A|E)S(B|E)$. 
This is always satisfied, unless we are in a trivial situation.
By symmetry, if Bob does first, still we have the same constraint 
(i.e.,~no constraint). Now consider the initial step is performed on $xK$ copies, 
where $K$ is large such that Theorem \ref{two-decouple} is applied. Because at any step of 
any round, the extraction randomness rate is the same equal to $I(A:B)$, 
so overall the rate is also $I(A:B)$. In addition with the local purity, the sum 
achieves the global purity. 

Now we compute the rate: denote $r:=\frac{S(A)S(B)}{S(A|E)S(B|E)}>1$. Suppose 
the ``ping-pong'' protocol ends at Alice's side after $L$ rounds with a large 
finite number $L$. Then Alice gets the randomness from $N_A$ copies 
where
\[\begin{split}
  N_A &= n_1+n_3+\cdots+n_{L+1}                   \\
      &= \frac{K}{S(A|E)}(1+r+r^2+\cdots+r^{L+1}) \\
      &= \frac{K}{S(A|E)}\frac{r^{L+2}-1}{r-1},
\end{split}\]
and Bob from
\[\begin{split}
  N_B &= n_2+n_4+\cdots+n_L                               \\
      &= \frac{KS(B)}{S(A|E)S(B|E)}(1+r+r^2+\cdots+r^{L}) \\
      &= \frac{KS(B)}{S(A|E)S(B|E)}\frac{r^{L+1}-1}{r-1},
\end{split}\]
and the total number of copies is $N=N_A+N_B$. The randomness rate depends only 
on the proportion of copies that randomness distiallion is performed on. So 
$\frac{N_B}{N_A}\approx\frac{S(B)}{rS(B|E)}=\frac{S(A|E)}{S(A)}$ and 
$\frac{N_A}{N}=\frac{S(A)}{I(A:B)}$ for large $L$. Thus we get  
$R_A=\frac{N_A}{N}I(A:B)+\log|A|-S(A)=\log|A|$ and $R_B=\log|B|-S(AB)$;
exchanging the role of Alice and Bob we get the rate pair $(\log|A|-S(AB),\log|B|)$. That the randomness generated at both sides is independent and secure against Eve comes from Theorem \ref{security}.

%%%%%%%%%%%%%%%%%%
\medskip\noindent
{\bf Proofs for the other settings.}

2) \emph{Free noise but no communication} means that we are effectively applying
the result in setting 1) to states $\rho_{AB}\otimes \frac{\1}{|A'|}\otimes \frac{\1}{|B'|}$ 
with large local dimension of $A'$ and $B'$ instead of $\rho_{AB}$. In this way,
the two conditional entropies $S(AA'|BB')$ and $S(BB'|AA')$ in setting 1) can be made positive while the local dimension increases,
hence the entire rate region becomes
\begin{align}
\label{eq:two-side:noise:0}
  R_A       &\leq \log|A| - S(A|B) , \nonumber\\
  R_B       &\leq \log|B| - S(B|A) , \\
  R_A + R_B &\leq R_G = \log|AB| - S(AB). \nonumber 
\end{align}

3) \emph{Free noise and free communication}: The rate region is 
characterised by 
%$R_A\leq R_G$, $R_B\leq R_G$, and $R_A+R_B\leq R_G$,
\begin{equation}
  \label{eq:all-free}
  R_A\leq R_G,\ R_B\leq R_G,\ \text{and}\ R_A+R_B\leq R_G,
\end{equation}
i.e.~the entire global purity can be realised on either side as
randomness (but not necessarily as purity!). To see this consider
an extreme point of the free-noise region, Eq.~(\ref{eq:two-side:noise:0}):
$R_A=\log|A|-S(AB)+S(B)=\log|A|-S(A)+I(A:B)$, $R_B=\log|B|-S(B)$. 
In this case, Bob's entire randomness, and the part $\log|A|-S(A)$ of Alice
are realised as purity obtained by first compressing their respective systems.
But purity can be exchanged freely via the dephasing channel, so 
the rate sum can be concentrated at Alice's side. Similarly for Bob.

4) For \emph{free communication but no noise} we have only an achievable
region from combining setting 1) with the sharing of purity
observed in setting 3): We can certainly achieve all rate pairs with
\begin{align}
  R_A       &\leq \log|A| - \max\{S(B),S(AB)\} , \nonumber\\
  R_B       &\leq \log|B| - \max\{S(A),S(AB)\} , \\
  R_A + R_B &\leq R_G = \log|AB| - S(AB). \nonumber
  \label{eq:two-side:0:comm}
\end{align}

%%%%%%%%%%%%%%%%%%%%
\medskip\noindent
{\bf Proof of tightness.}
We prove the tightness of the rate regions in three of these settings:
scenario 1), scenario 2), and scenario 3),
but leave 4) open.
In all settings, $R_A+R_B \le R_G = \log|AB|-S(AB)$, which clearly cannot
be beaten, even if Alice and Bob can freely cooperate. 
For setting 3): It clearly cannot be improved, because $R_G$ is reached 
for both parties, and nothing can be better. 
Setting 2): Since there is no communication, whatever Alice does, she cannot effect the purity $\log|B|-S(B)$ of Bob's state which can be used to  generate randomness. Hence $R_A \le R_G - (\log|B|-S(B)) = \log|A|-S(A|B)$. Similarly  
for Bob's rate. 
Setting 1): Even with free noise, $R_A \le \log|A|-S(A|B)$, so
this bound is still true without free noise. But also $R_A \le \log|A|$ always because
Alice has only her system $A$ to work on. So $R_A$ is bounded by 
the minimum of the two; likewise for Bob.
\end{proofthm}

%%%%%%%%%%%%%%%%%%%
\medskip
We arrive at the form of the one-sided irreducible ibit states by which 
we mean that all the randomness is in the key part at Alice's side and no extra randomness 
can be gained from the shield part. So any optimal protocol must reach 
this kind of state, asymptotically. 

{\lemma \label{irreducible-form}
Every irreducible ibit state with randomness only on Alice's side, is of the form
\[
\alpha_{K_AA'B'} = \frac{1}{|K_A|}\sum_{i,j=1}^{|K_A|} \ketbra{i}{j}_{K_A}
                                                 \otimes \frac{1}{|A'||B'|}U_iU_j^{\dagger},
\]
where $U_i$ are unitaries acting on $A'B'$.
}

\medskip
\begin{proof}
Suppose the purification for the irreducible state is $|\phi\>_{K_AA'B'E}$, 
where the randomness part is $K_A$ and $A'B'$ is the shield part. 
From the form of $\alpha_{K_AA'B'}$, we know
\[
  |\phi\>_{K_AA'B'E}=\frac{1}{\sqrt{|K_A|}}\sum_{i=1}^{|K_A|}\ket{i}_{K_A}\otimes U_{i}\ket{\phi_{0}}_{A'B'E}.
\]
Because it is an irreducible state which means the optimal rate of 
randomness that can be extracted under CLODCC is $\log |K_A|$, and so no protocol 
can produce more than $\log |K_A|$ bits of randomness. 
Now consider the Berta-Fawzi-Wehner protocol \cite{BFW},
which is naturally a CLODCC protocol where Bob does nothing. 
It creates $\log|K_A|+\log|A'|+S(AA'|E)_-$ bits which should be not larger 
than $\log |K_A|$, where we use the denotation $[t]_-=\min\{0,t\}$. Then we get $\log|A'|+S(K_AA'|E)_-\le 0$ that implies 
$S(K_AA'|E)=S(B')-S(E)\le 0$. So we have $\log|A'|+S(B')-S(E)\le 0$. 
Notice that $\alpha_{B'}=\frac{1}{|B'|}\1$ (or Bob can do local compression 
to get purity and send the purity to Alice), and 
$S(E)=S(A'B')_{\phi_i}\le S(A')_{\phi_i}+S(B')_{\phi_i}\le \log|A'|+\log|B'|$ for every $\ket{\phi_i}_{A'B'E}=U_{i}\ket{\phi_0}_{A'B'E}$. Then  the equality holds 
if and only if $\phi_{iA'B'}=\frac{1}{|A'|}\1_{A'}\otimes\frac{1}{|B'|}\1_{B'}$ for every $\ket{\phi_i}_{A'B'E}$, which means that $\ket{\phi_i}_{A'B'E}=U_{iE_1E_2}^{\top}\ket{\Phi}_{A'E_1}\otimes \ket{\Phi}_{B'E_2}=U_{iA'B'}\ket{\Phi}_{A'E_1}\otimes \ket{\Phi}_{B'E_2}$, where $U_{iE_1E_2}^{\top}$ is the transpose of $U_{iE_1E_2}$. Tracing out the $E$ part, we get the desired form.
\end{proof}

\medskip
From Lemma \ref{irreducible-form}, it is clear the last step of an optimal protocol 
is VQSM when Eve's system is included. 
The irreducible form shows that entropy of Eve is not less than Bob's entropy. 

%%%%%%%%%%%%%%%%%%%%%%%%
\medskip

\begin{proofthm}{\bf of Theorem \ref{randomness-one}.}
\newcommand{\RHS}{{\text{RHS}}}
Denote the quantity at the right-hand side as $\RHS$.

``$R_A\ge \RHS$'': For any CLODCC that transforms $n$ copies of $\rho_{AB}$ to 
a state $\sigma_{A'B'}$ whose purification is denoted as $\ket{\phi}_{A'B'E'}$. 
We notice that in the case of free communication and no noise in the two-side 
randomness scenario, Alice can get randomness at rate 
\[\begin{split}
  \log|A'B'|-&\max\{S(B'),S(E')\} \\
             &= n\log|AB|-\max\{S(B'),S(E')\},
\end{split}\] 
which holds because under CLODCC, Alice and Bob can only exchange
qubits, but the total number of their shared degrees of freedom remains 
constant.
So we have $R_A\ge \sup_n \left[\log|AB|-\max\left\{\frac1n S(B'),\frac1n S(E')\right\}\right]$,
which is the RHS.

``$R_A\le \RHS$'': Given $n$ copies of $\rho_{AB}$, a protocol to extract 
randomness ends in a state $\omega_{KA'B'}$ that is close to the state 
$\alpha_{KA'B'}=\frac1K \sum_{i,j=1}^K\ketbra{i}{j}_{K}\otimes U_i\sigma_{A'B'}U_j^{\dagger}$ 
in the sense that $\|\omega_{KA'B'}-\alpha_{KA'B'}\|_1 \leq \epsilon(n)\to 0$ (as $n\to\infty$). 
Now suppose the purifying state for $\omega_{KA'B'}$ is $\omega_{E'}$. 
From the continuity of entropy, we get 
$$|S(\omega_{KA'B'})-S(\alpha_{KA'B'})|\le \epsilon(n)\log|AB|^n+h(\epsilon(n)),$$
which amounts to
\begin{align*}
  |S(\omega_{E'})-S(\alpha_{E'})| &\leq \epsilon(n)\log|AB|^n+h(\epsilon(n)), \text{ and} \\
  |S(\omega_{B'})-S(\alpha_{B'})| &\leq \epsilon(n)\log|B'|+h(\epsilon(n)). 
\end{align*}
From the structure of $\alpha_{KA'B'}$, we know that 
$S(\alpha_{E'})=S(\sigma_{A'B'})\le\log|A'B'|$. Hence,
\[\begin{split}
  \log|K|&=\log|KA'B'|-\log|A'B'| \\
         &\leq n\log|AB|-S(\alpha_{E'}) \\
         &\leq n\log|AB|-S(\omega_{E'})+\epsilon(n)\log|AB|^n+h(\epsilon(n)).
\end{split}\]
Thus we have
$$\frac1n\log|K| \leq \log|AB|-\frac1n S(\omega_{E'})+\epsilon(n)\log|AB|+h(\epsilon(n)).$$ 
At the same time we have 
\[\begin{split}
  \frac1n\log|K| &\leq \log|AB|-\frac1n\log|B'| \\
                 &\leq \log|AB|-\frac1n S(\alpha_B') \\
                 &\leq \log|AB|-\frac1n S(\omega_{B'})+\epsilon(n)\log|AB|+h(\epsilon(n)).
\end{split}\]
Take the supremum over $n$, and notice that if for a finite $n$, the supremum is 
obtained, then for any $mn$, it is achieved. So we get $R_A=\sup_n \frac1n \log|K| \leq \RHS$.
\end{proofthm}

%%%%%%%%%%%%%%%%%%%%%%%%
\medskip
\begin{proofthm}{\bf of Theorem \ref{uppbound-one}.}
First we prove that the optimal randomness that can be extracted from $\rho_{AB}$ on Alice's 
side is upper bounded as
\[
R_A(\rho) \leq \log|AB|-\max\left\{\frac12[E_r^{\infty}(\rho_{AB})+S(AB)], S(AB)\right\}\!,
\]
where $E_r(\rho)=\inf_{\sigma}\tr\rho(\log\rho-\log\sigma)$, $\sigma$ running over separable states,
and $E^{\infty}_{r}(\rho)=\lim_{n\to\infty}\frac1nE_r(\rho^{\otimes n})$
are the relative entropy of entanglement \cite{Er} and its regularisation. Then from $E^{\infty}_{r}(\rho_{AB})\ge \max\{S(A)-S(AB),S(B)-S(AB)\}$ \cite{DevetakWinter-hash}, we get 
%$E^{\infty}_{r}(\rho_{AB})+S(\rho_{AB})\ge\max\{S(A),S(B)\}$. So 
the easy upper bound $R_A(\rho)\leq \log|AB|-\frac12\max\{S(A),S(B)\}$. 
%which is better than any upper bound given by entanglement measures with one half, 
%even the entanglement of purification \cite{ent-p} and entanglement of assistance \cite{ent-a} 
%given by $\log|AB|-\frac12\min\{S(A),S(B)\}$. 

It is straightforward to see that the randomness $R_A$ that can be extracted 
on Alice's side is not larger than the randomness extracted by global operation 
which is equal to the global purity. So $R_A(\rho)\leq\log|AB|-S(AB)$. Suppose 
that for $n$ copies of $\rho_{AB}$ whose purification is $\ket{\phi}_{ABE}$, 
with $E$ on Eve's side, we get an approximate irreducible state $\tilde{\alpha}_{KA'B'}$ 
under CLODCC operations, where $K$ is the randomness part and $A'B'$ is the shield 
part, $KA'$ is on Alice's side and $B'$ on Bob's, and its purification part is 
denoted as $E'$. Then,
\begin{align}
\log&|K| \le n\log|AB|-S(E'),\label{BFW-inq}\\
  &=   n\log|AB|-[S(E')-nS(E)]-nS(E),\label{product}\\
  &\le n\log|AB|-[E_r(\rho_{AB}^{\otimes n})-E_r(\tilde{\alpha}_{KA':B'})]-nS(E),\label{entr-incr}\\
  &\le n\log|AB|-E_r(\rho_{AB}^{\otimes n})-nS(E)+S(B'),\label{rel-ent}\\
  &\le n\log|AB|-E_r(\rho_{AB}^{\otimes n})-nS(E)+\log|B'|,\label{entr-dim}\\
  &=   n\log|AB|-E_r(\rho_{AB}^{\otimes n})-nS(E) \nonumber\\
  &\phantom{==========}
        +\log|KA'B'|-\log|KA'|,\\
  &\le 2n\log|AB|-E_r(\rho_{AB}^{\otimes n})-nS(AB)-\log |K|,\label{dim}
\end{align}
where Ineq. (\ref{BFW-inq}) comes from the Berta-Fawzi-Wehner protocol \cite{BFW}, Eq. (\ref{product}) from $nS(E)-nS(E)=0$, Eq. (\ref{entr-incr}) from Lemma \ref{key-ineq}, Ineq. (\ref{rel-ent}) from the fact that $E_r(\rho_{XY})\le S(\rho_Y)$ which follows from $\rho_{XY}=\tr_{X'}(\ketbra{\rho}{\rho}_{X'X:Y})$ where $\ket{\rho}_{X'XY}$ is a purification state for $\rho_{XY}$, and the nonincreasing property of entanglement under partial trace, Ineq. (\ref{entr-dim}) from $S(\rho_X)\le \log|X|$, and Ineq. (\ref{dim}) from that the dimension of $AB$ is invariant.  
So we obtain $\log |K|\le n\log|AB|-\frac{n}{2}[E^{\infty}_{r}(\rho_{AB})+S(\rho_{AB})]$. 
Since it holds for any randomness extraction, we 
get $R_A=\sup_n \frac1n \log |K|\le \log|AB|-\frac12 [E^{\infty}_{r}(\rho_{AB})+S(\rho_{AB})]$.

To prove the tightness for pure states, we first show it is tight for a singlet, that is $R_A=\frac{3}{2}$.
We get this by constructing an explicit protocol. Consider two singlets 
$\ket{\Phi}_{A_1B_1}\otimes \ket{\Phi}_{A_2B_2}$. Bob sends $B_1$ through 
the dephasing channel and the resulting state is 
$\frac{1}{\sqrt{2}}(\ket{000}\!+\!\ket{111})_{A_1B_1E}$ when we include Eve. 
Alice performs a CNOT to get $\ket{\Phi}_{A_1E}\otimes \ket{0}_{B_1}$,
where the pure state $\ket{0}_{B_1}$ can produce $1$ ibit. 
Then Alice performs a Bell measurement on $A_1A_2$ to get another
$2$ ibits, in total $3$ ibits from $2$ copies of the state. 
This is optimal from the upper bound for a singlet. The result can be generalised to a general pure state $\ket{\phi}_{AB}$
by noticing that the local purity $\log|B|-S(B)=\log|B|-S(A)$ at Bob's side can be transmitted freely through the dephasing channel. First both Alice and Bob turn the shared state into purity parts and information parts. Next Bob sends his purity part to Alice and Alice generates randomness from the purity parts. Then Alice and Bob extract randomness to Alice's side from the information part where the spectrum of the local state is nearly uniform by performing the similar protocol for a singlet. The total randomness comes from three parts, from the local purities of Alice and Bob at the rates $\log|A|-S(A)$ and $\log|B|-S(B)$ respectively, and from the information part at the rate of $\frac{3}{2}S(A)$. Noticing $S(A)=S(B)$ and summing up, we get the rate $R_A=\log|AB|-\frac12S(A)_{\phi}$.
\end{proofthm}

%%%%%%%%%%%%%%%%%%
\medskip
\begin{proofthm}{\bf of Theorem \ref{randomness-capacity}.}
Suppose for $n$ uses of the channel, the final tripartite state among Alice, Bob, and Eve is
\[
|\psi\>_{A^nB^nE^n}=V^{\otimes n}(|\phi\>_{A^nA'^n}),
\]
where $V: A'\to BE$ is an isometric dilation of the channel ${\cal N}$. 
Then with free noise, from scenario 2) of Theorem \ref{two-side} we know that
for such a state, Bob can extract randomness at rate 
$n\log|B|-S(B^n|A^n)_{\psi_{A^nB^n}}=n\log |B|+[S(A^n)-S(A^nB^n)]_{\psi_{A^nB^n}}$. 
Further, from Theorem \ref{two-decouple}, we know that the randomness value 
which is represented by the dephased system $K_B$ is decoupled from both 
system A of the sender, and E of the eavesdropper, thus is private to each 
of them individually. 

To maximize the quantity amounts to maximizing the second term, which is 
none other than the reverse coherent information, and from \cite{DJKR,reverse-coh} we 
know that the reverse coherent information of a channel is additive. So optimisation over $n$ channel uses is reduced to the optimization over a single channel use. 

For the sake of completeness, we write the proof. We will prove that 
\[\begin{split}
  \max_{|\phi\>_{A_1A_2A'_1A'_2}}[S(\phi_{A_1A_2})&-S({\cal N}_1\otimes{\cal N}_2(\phi_{A_1A_2A'_1A'_2}))] \\
     &=  \max_{|\psi\>_{A_1A'_1}}[S(\psi_{A_1})-S({\cal N}_1(\psi_{A_1A'_1}))] \\
     &\phantom{=}
         + \max_{|\varphi\>_{A_2A'_2}}[S(\varphi_A)-S({\cal N}_2(\varphi_{A_2A'_2}))],
\end{split}\]
for quantum channels ${\cal N}_1: A'_1\to B_1$ and ${\cal N}_2: A'_2\to B_2$.

It is easy to see the direction ``$\ge$'' holds by taking $|\phi\>_{A_1A_2A'_1A'_2}=|\psi\>_{A_1A'_1}\otimes|\varphi\>_{A_2A'_2}$. 

The nontrivial part is the opposite direction ``$\le$''. First we give another equivalent form of the optimisation. Let $V: A'\to BE$ be an isometric dilation of the channel ${\cal N}: A'\to B$. Then for an input $\ket{\phi}_{AA'}$, we have 
\[
\1_A\otimes V\ket{\phi}_{AA'}=\ket{\psi}_{ABE}.
\]
From the duality relation between conditional entropies of a pure state $\ket{\psi}_{ABE}$ 
\[
  S(A)_{\psi}-S(AB)_{\psi}=-S(B|A)_{\psi}=S(B|E)_{\psi},
\] 
and $\phi_A=\psi_A$, $\psi_{AB}={\cal N}(\phi_{AA'})$, we have
\be
\label{dual}
S(\phi_{A})-S({\cal N}(\phi_{AA'})) =S(B|E)_{\psi},
\ee
where $\psi_{BE}=V\phi_{A'}V^{\dagger}$. In Eq. (\ref{dual}), the optimisation over the input $\ket{\phi}_{AA'}$ on the left side is reduced to the optimisation over the mixed state $\phi_{A'}$ on the right side. So we get the equivalent form of RCI,
\[
\max_{\ket{\phi}_{AA'}}[S(\phi_{A})-S({\cal N}(\phi_{AA'}))]=\max_{\phi_{A'}}S(B|E)_{\psi}.
\]
 
Given two channels ${\cal N}_1$ and ${\cal N}_2$, denote their isometries as $V_1: A'_1\to B_1E_1$ and $V_2: A'_2\to B_2E_2$ respectively. Suppose the optimal input state for the two channels is $\sigma_{A'_1A'_2}$ and the output on $B_1E_1B_2E_2$ is
\[
\rho_{B_1E_1B_2E_2}=V_1\otimes V_2\sigma_{A'_1A'_2}V^{\dagger}_1\otimes V^{\dagger}_2.
\] 

By the chain rule $S(XY|Z)=S(X|Z)+S(Y|XZ)$ and the monotonicity of 
conditional entropy (coming from the strong subadditivity) $S(X|YZ) \leq S(X|Y)$, we get
\begin{align*}
  S(B_1B_2|E_1E_2)_{\rho}&=   S(B_1|E_1E_2)_{\rho}+S(B_2|B_1E_1E_2)_{\rho}, \\
                  &\le S(B_1|E_1)_{\rho}+S(B_2|E_2)_{\rho},
\end{align*}
where $\rho_{B_1E_1}=V_1\sigma_{A'_1}V^{\dagger}_1$, $\rho_{B_2E_2}=V_2\sigma_{A'_2}V^{\dagger}_2$. So the direction ``$\le$''  is true and we get the result of additivity.  

\medskip

Next we show that the function $S(B|E)$ on $\rho_{BE}=V\sigma_{A'}V^{\dagger}$ is concave with respect to the input state $\sigma_{A'}$, thus its optimisation problem is efficiently computable. We will prove that for the input state $\sigma=p\sigma_1+(1-p)\sigma_2$ with $0\le p\le 1$, and output
\[
\rho_{BE}=p\rho_1+(1-p)\rho_2=pV\sigma_1V^{\dagger}+(1-p)V\sigma_2V^{\dagger},
\]
we have 
\be\label{concave}
S(B|E)_{\rho}\ge pS(B|E)_{\rho_1}+(1-p)S(B|E)_{\rho_2}.
\ee

To that end, consider the state 
\[
  \rho_{BEF} = p\rho_1\otimes\ketbra{0}{0}_{F}+(1-p)\rho_{2}\otimes\ketbra{1}{1}_{F}.
\] 
By the monotonicity of the conditional entropy $S(B|E)\ge S(B|EF)$ and the computation of  $S(B|E)$ and $S(B|EF)$, we arrive at Ineq. (\ref{concave}).
\end{proofthm}

\end{document}